\documentclass[11pt]{article}
\pdfoutput=1
\usepackage{tensor, physics, caption, subcaption} 
\usepackage[dvipsnames]{xcolor}
\usepackage{amsfonts}
\usepackage{amssymb}
\usepackage{mathrsfs}
\usepackage{graphicx}
\usepackage{amsmath,amsthm,amssymb,amscd}
\usepackage[all]{xy}
\usepackage{multirow}
\usepackage{color}
\usepackage{float}
\usepackage{mathtools,slashed,tensor}
\usepackage{bbold}
\usepackage{bbm}
\usepackage{slashed}
\usepackage{cancel}
\usepackage[normalem]{ulem}
\usepackage{cite, comment}
\usepackage{mdframed}

\newtheorem{lemma}{Lemma}
\newtheorem{theorem}{Theorem}

\theoremstyle{definition}
\newtheorem{definition}{Definition}
\newtheorem{assumption}{Assumption}
\newtheorem{condition}{Condition}

\theoremstyle{remark}
\newtheorem{remark}{Remark}
\usepackage{jheppub}

\def\be{\begin{equation}}
\def\ee{\end{equation}}
\def\bea{\begin{eqnarray}}
\def\eea{\end{eqnarray}}

\preprint{LITP-25-16}

\title{A Universality Theorem for the Quantum Thermodynamics of Near-Extremal Black Holes }

\author[a]{Leopoldo A.~Pando Zayas}

\author[a]{ and Jingchao Zhang}

\emailAdd{ lpandoz@umich.edu, jingchaz@umich.edu }

\affiliation[a]{Leinweber Institute for Theoretical Physics, 
University of Michigan, Ann Arbor, MI 48109, USA}

\abstract{We prove that the one-loop contribution from tensor modes to the thermodynamic entropy of near-extremal black holes is universal. Our proof applies to asymptotically flat, Anti-de-Sitter and de-Sitter black holes; it also covers spherical, axial and planar symmetries. We consider black hole configurations with and without matter sectors and explicitly discuss  Abelian gauge fields and neutral scalar fields with arbitrary potential. We demonstrate that under certain  conditions, the thermodynamics of near-extremal black holes  contains a one-loop contribution from the tensor modes that equals $\frac{3}{2}\log (T_{\rm Hawking}/T_q)$. The proof of this theorem also shows explicitly how the Schwarzian modes appear universally in near-extremal geometries in dimensions four, five and six. We apply this theorem to Kerr-de-Sitter black holes as an explicit example.}

\begin{document}
\today 
\maketitle

\newpage
\section{Introduction}

Black hole thermodynamics is one of the most impressive results in the classical physics of black holes, it opens a door to quantum aspects of gravity by querying a microscopic explanation of the entropy. One of the most significant achievements of string theory resides precisely in providing a microscopic explanation for the entropy of certain black holes \cite{Strominger:1996sh}. More recently, in the context of the AdS/CFT correspondence, large classes of asymptotically AdS black holes have been given a statistical interpretation in terms of counting certain operators on the field theory dual \cite{Benini:2015eyy,Cabo-Bizet:2018ehj, Choi:2018hmj, Benini:2018ywd}. 

There are various approaches to establishing the thermodynamic properties of black holes. One possibility is to start from the definition of conserved charges (mass, angular momentum and electric charge)  and to establish thermodynamics relation as a statement of probe matter falling behind the horizon leading to Smarr relations \cite{Wald:1984rg}. One can also arrive at black hole thermodynamics by considering the properties of two nearby solutions to  Einstein equations with conserved charges $M+\delta M, J+\delta J$  \cite{Townsend:1997ku}. Most relevant for use is the approach based on the gravitational path integral with appropriate boundary conditions as originally formulated by Gibbons and Hawking \cite{Gibbons:1976ue}. In this manuscript we will use the Euclidean path integral approach including its one-loop order leading to a universal quantum result for the thermodynamics of near-extremal black holes. 

For near-extremal black holes, it has been long known that their thermodynamics must breakdown  at very low temperatures \cite{Preskill:1991tb} (see also \cite{Maldacena:1998uz,Page:2000dk}). Recently, a thermodynamic correction, first understood in the context of the Jackiw-Teitelboim (JT) two-dimensional quantum gravity, has helped elucidate the thermodynamics corrections that are central to near-extremal black holes. A key role in correcting the thermodynamics of near-extremal black holes is played by certain zero modes in the long near-AdS$_2$ throat. The systematic treatment of such zero modes has been shown to affect the thermodynamics of near-extremal black holes \cite{Iliesiu:2020qvm,Heydeman:2020hhw,Boruch:2022tno,Iliesiu:2022onk}, including also in rotating black holes where a dimensional reduction to an effective two-dimensional theory is not available \cite{Kapec:2023ruw,Rakic:2023vhv,Maulik:2024dwq}.

The study of this particular type of gravitational quantum corrections is one of the few situations where a controlled approximation exists that does not require a full quantum theory of gravity; it suffices to carefully track the role of a few modes. Naturally, the implications of these quantum fluctuations in the throat have been investigated beyond thermodynamics and into the realm of more dynamical questions such as their influence in Hawking radiation~\cite{Brown:2024ajk,Maulik:2025hax,Lin:2025wof,Bhattacharjee:2025wfv,Rakic:2025svg,Luo:2026epp} and broader aspects of the quantum cross section~\cite{Emparan:2025sao,Biggs:2025nzs,Emparan:2025qqf,Betzios:2025sct,Jiang:2025cyl}; there have also been some studies exploring implications for aspects of the AdS/CFT correspondence~\cite{Daguerre:2023cyx,Liu:2024gxr,Liu:2024qnh,Nian:2025oei,PandoZayas:2025snm,Cremonini:2025yqe,Gouteraux:2025exs,Kanargias:2025vul}.

Through the lens of the gravitational path integral we show that, as in previous cases, the naive evaluation is plagued by infrared divergencies due to the presence of tensor, vector and gauge zero modes. One is required to first regularize these divergences by turning on a small temperature and subsequently analyze the imprint of such small temperature in the physics. Given the ubiquity of this operation, in this manuscript we aim to distill the results of various works into a {\it theorem}, extending key observations of \cite{Maulik:2024dwq}. As an explicit new example, we consider the  quantum-corrected thermodynamics of rotating asymptotically de-Sitter spacetime. Throughout we emphasize how the treatment changes from the simple spherically symmetric cases to axial-symmetric ones which are   applicable to rotating near-extremal black holes.   
    
The manuscript is organized as follows. In Section~\ref{Sec:A brief review to the procedure}, we briefly review the procedure of Euclidean path integral for near-extremal black holes. and set up the theory that we will consider in the rest of the paper. In Section~\ref{Sec:Ansatz:}, we provide two lemmas that help characterize the general form of near-horizon extremal configurations. We impose an empirical assumption of the small temperature perturbation of the metric and matter fields. We prove the existence of the tensor zero modes as another lemma in Section~\ref{Sec: Tensor zero modes} and show  how the higher dimensional tensor zero modes relate to the boundary curve reparameterization modes in the framework JT gravity. In Section~\ref{Sec: cancellation}, we prove that only two terms in the Lichnerowicz operator contribute to the lifted eigenvalues of the would-be zero modes. We show that, in four, five and six dimensions, the lifted eigenvalues always contribute $\frac{3}{2}\log{T_{\rm Hawking}}$  to the thermodynamic entropy. We also comment on further generalizations to higher dimensions. We consider planar symmetric solutions in Section~\ref{Sec: Planar}. In Section~\ref{Sec: Kerr_dS}, we provide an explicit example of the application of our theorem to the case of rotating asymptotically de-Sitter black holes, and comment on some peculiarities of asymptotic de-Sitter black holes.

\section{Review of one-loop temperature corrections}
\label{Sec:A brief review to the procedure}
In this section we briefly review the Euclidean path integral evaluation that has been implemented in a number of  explicit examples of near-extremal black holes \cite{Kapec:2023ruw,Rakic:2023vhv,Maulik:2024dwq,Maulik:2025phe,Blacker:2025zca}. We focus on the cases of dimension four and above, some discussion of the three-dimensional case can be found in \cite{Ghosh:2019rcj} which our approach covers as a particular case. Our goal is to take the universal result gleamed in \cite{Maulik:2024dwq}, review the key ingredients, and adapt this procedure to a general setup, therefore establishing a theorem of the logarithmic temperature correction to the extremal black hole entropy.

The thermodynamic entropy of a black hole can be obtained from the thermal partition function through the Euclidean path integral \cite{Gibbons:1976ue}: 
\begin{equation}
    Z=\int[\mathcal{D}g][\mathcal{D}\Phi]e^{-I[g,\Phi]},
\end{equation}
where $g$ is the metric field, while $\Phi$ represents any possible matter fields. We use $I$ to denote the Euclidean action
\begin{equation}
    I=I_{bulk}+I_{boundary}+I_{gauge-fixing}.
\end{equation}
By using saddle-point approximation
\begin{equation}
    Z\approx Z_0=e^{-I[\bar{g},\bar{\Phi}]},
\end{equation}
we find the well-known thermodynamic entropy of a black hole that is proportional to the area of its horizon
\begin{equation}
    S_0= \frac{A_h}{4}.
\end{equation}
In most treatable cases, the saddle-point approximation works well because loop corrections are suppressed by Newton's constant $G_N$. For near-extremal black holes this paradigm breaks down due to the presence of certain would-be zero modes that become strongly coupled in the $T\to 0$ limit. 

Consider small fluctuations around a classical saddle solution $(\bar{g}, \bar{\Phi})$:
\begin{equation}
    g=\bar{g}+h,\quad \Phi=\bar{\Phi}+\phi.
\end{equation}
The partition function, up to one loop order, is given by 
\begin{equation}
\label{Eq:GPI-1L}
    Z=Z_0\int[\mathcal{D}h][\mathcal{D}\phi]e^{-\frac{1}{2}h^*\,\Delta_L\,h-\frac{1}{2}\phi^*\,\Delta_M\,\phi},
\end{equation}
where $\Delta_L,\, \Delta_M$ are the corresponding quadratic operators. For near-extremal black holes, one finds that there are modes with eigenvalues that vanish as one decreases the temperature of the background black hole configuration
\begin{equation}
\lim_{T\to0}\Delta_Lh_{zero}=0
\end{equation}
Beyond these graviton zero modes, there could also be matter field zero modes. For example, there is a careful discussion of the role of gauge (photon) zero modes \cite{Blacker:2025zca}. In this manuscript we are interested in universal aspects and focus only on the one-loop correction from the metric fluctuations. The Gaussian integration of graviton fluctuations in \eqref{Eq:GPI-1L} gives
\begin{equation}
\label{Eq: one-loop formula}
    \int[\mathcal{D}h]e^{-\frac{1}{2}h^*\,\Delta_L\,h}=(\det \Delta_L)^{-\frac{1}{2}}.
\end{equation}
The above expression diverges if the determinant $\Delta_L$ (given as the product of all eigenvalues) admits a vanishing eigenvalue. In other words, the quantum correction becomes dominant when the temperature is small enough, implying that the  semiclassical picture of black hole thermodynamics breaks down for extremal black holes. This breakdown was anticipate in \cite{Preskill:1991tb,Maldacena:1998uz,Page:2000dk} by explicitly showing that the semiclassical picture breaks down at small temperatures and identifying the corresponding scale. The proper treatment of these zero modes has been elaborated in a series of works \cite{Iliesiu:2022onk,Banerjee:2023quv}, for Kerr black holes \cite{Kapec:2023ruw,Rakic:2023vhv,Maulik:2024dwq,Arnaudo:2024rhv}, and more recently for dS \cite{Maulik:2025phe,Blacker:2025zca,Arnaudo:2025btb}. To extract how the thermodynamics changes when the temperature is small enough, we should use a small-temperature black hole as the background configuration
\begin{equation}
g=\bar{g}+T\delta g\quad\Phi=\bar{\Phi}+T \delta\Phi, 
\end{equation}
where \{$\bar{g}$,$\bar{\Phi}$\} is the extremal configuration. By turning on a small temperature, one finds that the ``zero modes"  now have eigenvalues lifted from zero at order $T$:
\begin{equation}
    \Delta_L h_{zero}=0+\delta \Lambda\; h_{zero},\quad\delta\Lambda=\tilde{\Lambda}\;T+\mathcal{O}(T^2).
\end{equation}
Keeping the leading order of $\delta\Lambda$, after proper regularization, one finds \eqref{Eq: one-loop formula} gives a one-loop correction to the partition function 
\begin{equation}
    Z\sim Z_0\;T^{\frac{3}{2}},
\end{equation}
which leads to a correction to the entropy of 
\begin{equation}
    \delta S=\frac{3}{2}\log{T}.
\end{equation}
Matter zero modes can also contribute to the quantum correction of entropy. As the main purpose of this paper is to show the universality of the $\frac{3}{2}\log{T}$ correction from the metric tensor zero modes, we do not consider matter field fluctuations from now on.

Based on the accumulated evidence from various classes of examples, in the rest of the paper, we will perform a general analysis and present lemmas and a theorem to explain the generality of this $\tfrac{3}{2}\log{T}$ correction and to elucidate the reasons behind this universality.

Before going into the next section, let us first set up the theory we consider. In this paper, we will only consider black holes in $D\ge4$ dimensional spacetime. We work with a general second derivative theory describing Einstein gravity coupled to Abelian vectors $A^I\,(I=1\dots N)$ and uncharged scalars $\phi^A (A=1\dots M)$ with the Euclidean action
\begin{equation}
\label{Eq: Maxwell action}
    I_{bulk}=-\frac{1}{16\pi G_N} \int_{\mathcal{M}} d^D x \sqrt{-g}\,\left(R-\frac{1}{2}f_{AB}(\phi)\partial_{\mu}\phi^A\partial^{\mu}\phi^B- V(\phi)- \, \mathcal{C}_{IJ}(\phi)F^{I}_{\mu\nu}F^{J\mu\nu}\right),
\end{equation}
where $F^I\equiv dA^I$ is the field strength. An arbitrary scalar potential is given by $V(\phi)$, which allows for a cosmological constant. As we only consider metric fluctuation in this paper, we only need a gauge-fixing term for $h$
\begin{equation}
\label{Eq: gauge fixing}
    I_{gauge-fixing}= \frac{1}{32\pi G_N} \, \int_\mathcal{M}d^nx\,\bar{g}_{\mu\nu} 
\left( \bar{\nabla}_\alpha h^{\alpha\mu} - \frac{1}{2} \bar{\nabla}^\mu h^\alpha_{\ \alpha} \right)
\left( \bar{\nabla}_\beta h^{\beta\nu} - \frac{1}{2} \bar{\nabla}^\nu h^\beta_{\ \beta} \right),
\end{equation}
where the bars indicate that the covariant derivative operators are compatible with the extremal metric $\bar{g}$. This fixes metric fluctuation in the ``traceless and divergence free gauge"
\begin{equation}
\label{Eq: gauge}
    \bar{h}^\mu_{\,\,\mu}= h_{\mu\nu}\bar{g}^{\mu\nu}=0,\quad\bar{\nabla}^\mu h_{\mu\nu}=0.
\end{equation}
For convenience of future use, we quote the Lichnerowicz operator $\Delta_L$ here (more details can be found in  \cite{Maulik:2024dwq}),
\begin{equation}
\label{Eq: Linchnerowichz operator}
    h_{\alpha\beta}^*\Delta^{\alpha \beta, \mu \nu}_{L}h_{\mu\nu}=-\frac{1}{16\pi G}h_{\alpha\beta}^*\left(\Delta^{\alpha \beta, \mu \nu}_{EH}-V(\phi)\,\Delta^{\alpha \beta, \mu \nu}_{1}-\Delta^{\alpha\beta,\mu\nu}_\phi-2\Delta^{\alpha \beta, \mu \nu}_{F}\right)h_{\mu\nu},
\end{equation}
where
\begin{align}
\nonumber
    h_{\alpha\beta}^*\Delta^{\alpha \beta, \mu \nu}_{EH}h_{\mu\nu}=& h_{\alpha \beta}^*(\frac{1}{2} \bar{g}^{\alpha \mu} \bar{g}^{\beta \nu} \bar{\square}-\frac{1}{4} \bar{g}^{\alpha \beta} \bar{g}^{\mu \nu} \bar{\square}+ \bar{R}^{\alpha \mu \beta \nu}\\
&+ \bar{R}^{\alpha \mu} \bar{g}^{\beta \nu}-\bar{R}^{\alpha \beta} \bar{g}^{\mu \nu}-\frac{1}{2} \bar{R} \bar{g}^{\alpha \mu} \bar{g}^{\beta \nu}+\frac{1}{4} \bar{R} \bar{g}^{\alpha \beta} \bar{g}^{\mu \nu}) h_{\mu \nu},\\
h_{\alpha \beta}^* \Delta_{1}^{\alpha \beta, \mu \nu} h_{\mu \nu}=&h_{\alpha\beta}^*(-\frac{1}{2}g^{\alpha\mu}g^{\beta\nu}+\frac{1}{4}g^{\alpha\beta}g^{\mu\nu})h_{\mu\nu},\\
\nonumber
h_{\alpha \beta}^* \Delta_{F}^{\alpha \beta, \mu \nu} h_{\mu \nu}=&h_{\alpha \beta}^*\, \mathcal{C}_{IJ}(\phi)\Bigl( -\frac{1}{8} F^I_{\rho\sigma}F^{J\rho\sigma}\left(2 g^{\alpha \mu} g^{\beta \nu}-g^{\alpha \beta} g^{\mu \nu}\right)+ F^{I\alpha \mu} F^{J\beta \nu}\\
&+2 F^{I\alpha \gamma} F^{J\mu}{}_{\gamma} g^{\beta \nu}-F^{I\alpha \gamma} F^{J\beta}{}_{\gamma} g^{\mu \nu} \Bigr) h_{\mu \nu},\\
\nonumber
h_{\alpha \beta}^* \Delta_{\phi}^{\alpha \beta, \mu \nu} h_{\mu \nu}=&h_{\alpha \beta}^*\,f_{AB}(\phi)\Bigl(-\frac{1}{8}\partial_\rho\phi^A\partial^{\rho}\phi^B\left(2g^{\alpha\mu}g^{\beta\nu}-g^{\alpha\beta}g^{\mu\nu}\right)\\
&+g^{\beta\nu}\partial^{\alpha}\phi^{A}\partial^{\mu}\phi^B-\frac{1}{2}g^{\mu\nu}\partial^{\alpha}\phi^{A}\partial^{\beta}\phi^B\Bigr)h_{\mu \nu}.
\end{align} 




\section{The near-horizon extremal configuration}
\label{Sec:Ansatz:}
The quantum correction to the thermodynamics of near-extremal black holes is generated by strong quantum fluctuations in the near-horizon region of the black holes. In this section, we firstly set up a general form of the near-horizon configuration, including the metric and matter fields for which will closely follow \cite{Kunduri:2008rs} and implement a direct generalization to higher dimensions. Secondly, we propose an empirical assumption of the small temperature perturbation of the metric and the gauge field, which will serve as a regulator for the tensor zero modes in Section \ref{Sec: cancellation}.

\subsection{The background metric}
In this subsection, we investigate the general form of the near-horizon geometry of near-extremal black holes. We draw on results described in Appendix A of \cite{Friedrich:1998wq}. In a neighborhood of a Killing horizon of a black hole, one can introduce Gaussian null coordinates $\{v, r, x^a\}$, in which the metric takes the form of \cite{Friedrich:1998wq}
\begin{equation}
\label{Eq: near-horizon non-extremal}
    ds^{2}
  = r\, \mathcal{F}(r,x)\, dv^{2}
  + 2\, dv\, dr
  + 2 r\, h_{a}(r,x)\, dv\, dx^{a}
  + \gamma_{ab}(r,x)\, dx^{a} dx^{b},
\end{equation}
where $\mathcal{F},\,h_a,\,\gamma_{ab}$ are smooth functions in $r,\,x$. The sub-indices $a,b$ go from $1$ to $D-2$. The horizon corresponds to $r=0$. 
\begin{definition}
\label{def: 1}
We define a black hole to be at its \textit{quadratic extremality} if: (1) it has a Killing horizon; (2) in the neighborhood of its Killing horizon the metric \eqref{Eq: near-horizon non-extremal} satisfies that $\mathcal{F}(r,x)=r \,\tilde{\mathcal{F}}(r,x)$ where $\tilde{\mathcal{F}}$ is a smooth function and $\tilde{\mathcal{F}}(0,x)\neq0$.
\end{definition}

\begin{remark}
\label{rem: 1}
The above definition is equivalent to $g_{vv}\sim r^2$ as $r\to0$. This is the geometric behavior of the metric when two of its horizons coincide.
\end{remark}

Given the above definition, we have that the metric in the neighborhood of a Killing horizon of a black hole at its quadratic extremality takes the form of 
\begin{equation}
    ds^{2}
  = r^2\, \tilde{\mathcal{F}}(r,x)\, dv^{2}
  + 2\, dv\, dr
  + 2 \,r\, h_{a}(r,x)\, dv\, dx^{a}
  + \gamma_{ab}(r,x)\, dx^{a} dx^{b}.
\end{equation}
By taking the near-horizon limit
\begin{equation}
\label{Eq: near-horizon limit}
    v\to\frac{v}{\epsilon},\qquad r\to\epsilon \,r,
\end{equation}
and $\epsilon\to0$, one gets
\begin{equation}
    ds^{2}
  = r^2\, \tilde{\mathcal{F}}(x)\, dv^{2}
  + 2\, dv\, dr
  + 2\, r\, h_{a}(x)\, dv\, dx^{a}
  + \gamma_{ab}(x)\, dx^{a} dx^{b},
\end{equation}
where $\tilde{\mathcal{F}}(x)\equiv\tilde{\mathcal{F}}(0,x),\,h_a(x)\equiv h(0,x),\, \gamma_{ab}(x)\equiv\gamma_{ab}(0,x)$.
Then we impose the following additional assumptions to the black hole:

\begin{assumption}
\label{ass: 1}
The black hole has $\lfloor\frac{D-1}{2}\rfloor$ rotational symmetries $U(1)^{\lfloor\frac{D-1}{2}\rfloor}$, where $\lfloor\frac{D-1}{2}\rfloor$ is the maximum number of independent rotation directions a black hole can admit. 
\end{assumption}
Given the additional assumptions above, we are ready to generalize Lemma 1 in \cite{Kunduri:2008rs} to higher dimensions. First, given the rotational symmetries, we can separate $x^a$ into $\lfloor\frac{D-1}{2}\rfloor$ rotational directions $\varphi_i$ and $D-2-\lfloor\frac{D-1}{2}\rfloor$ non-rotational directions $\theta_n$. Then we can rewrite the metric as 
\begin{equation}
\begin{aligned}
ds^{2}
  =&r^2\, \tilde{\mathcal{F}}(\theta_n)\, dv^{2}
  + 2\, dv\, dr
  + 2\, r\, h_{m}(\theta_n)\, dv\, d\theta_m+2\,r\,h_{i}(\theta_n)dv\,d\varphi_i\\
  &+g_{mm}(\theta_n)d\theta_m^2+g_{ij}(\theta_n)d\varphi_id\varphi_j,
\end{aligned}
\end{equation}
where the indices $m,n$ are used for $\theta$ directions, while $i,j$ are used for $\varphi$ directions.
We define a positive function  $\Gamma(\theta_n)$ by
\begin{equation}
    h_{n}=-\frac{\partial\Gamma/\partial\theta_n}{\Gamma},
\end{equation}
and functions $k_i(\theta_n)$ by
\begin{equation}
    h_{i}=\Gamma^{-1}k_i.
\end{equation}
Performing a further coordinate transformation $r\to\Gamma \,r$ leads to 
\begin{equation}
\label{Eq: g bar after Gamma}
    ds^{2} = r^{2} A(\theta_n) dv^{2} + 2 \Gamma(\theta_n) dv\,dr + g_{mm}(\theta_n)d\theta_m^{2} 
+ g_{ij}(\theta_n) \Bigl( d\varphi_i + k^{i}(\theta_n) r dv \Bigr) 
\Bigl( d\varphi_j + k^{j}(\theta_n) r dv \Bigr),
\end{equation}
where $g_{ij}\,k^i=k_j$, $A=g_{ij}\,k^ik^j+\Gamma^2\tilde{\mathcal{F}}$. 
Then, we further impose the following additional condition:
\begin{condition}
\label{con: 1}
The $\theta_m\, \varphi_i$ and $\theta_m\, v$ components of the Ricci tensor vanish.
\end{condition}
By direct evaluation, we find \textit{Condition \ref{con: 1}} implies
\begin{equation}
R_{mi} = \frac{1}{2 \Gamma} g_{ij} \frac{\partial}{\partial\theta_m}k^{j}=0,
\end{equation}
\begin{equation}
    R_{m v} = \frac{r}{\Gamma} \left[ \frac{\partial A}{\partial\theta_m} -\frac{\partial\Gamma}{\partial\theta_m} \frac{A}{\Gamma}  +  \frac{\partial k^{i}}{\partial\theta_m}  k_{i} \right]=0,
\end{equation}
which further implies $k^i$ are constants and $A(\theta_n)=A_0\Gamma(\theta_n)$, where $A_0$ is a constant. Now \eqref{Eq: g bar after Gamma} can be simplified to 
\begin{equation}
\label{Eq: gbar Kunduri}
ds^{2} =  \Gamma(\theta_n)[A_0 r^{2} dv^{2} + 2 \,dv\,dr] + g_{mm}(\theta_n)d\theta_m^{2} 
+ g_{ij}(\theta_n) \Bigl( d\varphi_i + k^{i} r dv \Bigr) 
\Bigl( d\varphi_j + k^{j} r dv \Bigr).
\end{equation}


In appendix \ref{App: Coordinate transformation}, we perform further coordinate transformations that allow for a direct comparison with the backgrounds discussed in \cite{Maulik:2024dwq}. 

Then we conclude the statement above with the following lemma:
\begin{lemma}
\label{lem: 1}
\textbf{(Extended Kunduri-Lucietti-Reall Lemma 1)}
Assume that a black hole solution of the theory given by the action in \eqref{Eq: Maxwell action} is at quadratic extremality and satisfies \textit{Assumption \ref{ass: 1}}, the Euclidean near-horizon metric takes the form of 
\begin{equation}
\label{Eq: BG ansatz}
\begin{aligned}
ds^2=&g_1(\theta_n)\!\left[(y^{2}-1)\,d\tau^{2}+\frac{dy^{2}}{y^{2}-1}\right]\;+\;g_{mm}(\theta_n)\,d\theta_m^{2}\;\\
&+g_{ij}(\theta_n)\!\left(d\varphi_i + i\,k_i\,(y-1)\,d\tau\right)\left(d\varphi_j + i\,k_j\,(y-1)\,d\tau\right),
\end{aligned}
\end{equation}
where repeated $\{i,j\}$ follows the Einstein summation convention while $g_{mm}d\theta_m^2$ term only sums once over $m$. 
\end{lemma}
\begin{remark}
\label{rem: 2}
    In the above proof, we have assumed \textit{Condition \ref{con: 1}} is true. However, when we prove \textit{Lemma \ref{lem: 2}} in the next subsection, we will see \textit{Condition \ref{con: 1}} can be proved from \textit{Assumption \ref{ass: 1}} for any black hole solution of the theory \eqref{Eq: Maxwell action} at \textit{quadratic extremality}. Therefore, we do not list \textit{Condition \ref{con: 1}} as an independent assumption in the above lemma.
\end{remark}

\begin{remark}
\label{rem: 3}
    Inside the square brackets of \eqref{Eq: BG ansatz} is the Euclidean $AdS_2$ expected from the near-horizon geometry of extremal black holes. The metric function $g_1(\theta_n)$ describes how the $AdS_2$ is warped when the black hole is axial symmetric. Each $k_i$ is related to the angular velocity along $\varphi_i$ at the horizon. In spherically symmetric cases, such as  Reissner--Nordstr\"om black holes, the functions $g_1,\,g_{mm},\,g_{ij}$ degenerate to constants independent of $\theta_n$, and all $k_i$ vanish.
\end{remark}

It is instructive to explicitly display the near-horizon extremal geometry in four dimensions:
\begin{equation}
\label{Eq: BG ansatz 4d}
ds^2=g_1(\theta)\!\left[(y^{2}-1)\,d\tau^{2}+\frac{dy^{2}}{y^{2}-1}\right]\;+\;g_2(\theta)\,d\theta^{2}+g_{3}(\theta)\!\left(d\varphi + i\,k\,(y-1)\,d\tau\right)^2.
\end{equation}
Information of the black hole asymptotics, charges, angular momenta, etc., is contained in the functions $g_1(\theta),\,g_2(\theta),\,g_3(\theta)$ and constant $k$.

\subsection{The matter sector}
For a plethora of cases, it was found in \cite{Maulik:2024dwq} that the one-loop contribution  to the Euclidean path integral from the tensor zero modes is insensitive to the matter sector.  To prove this more generally, we need a general form of the matter fields.


The general form of a near-horizon $U(1)$ gauge field in four and five dimensions is given by Section 2.4 of \cite{Kunduri:2008rs}. This argument can be naturally generalized to higher dimensions. We introduce the following extended lemma:
\begin{lemma}
\label{lem: 2}
\textbf{(Extended Kunduri-Lucietti-Reall Section 2.4)}
    Assume that a black hole solution of the theory given by \eqref{Eq: Maxwell action} is at quadratic extremality and satisfies \textit{Assumption \ref{ass: 1}}, the near-horizon limits of a $U(1)$ gauge field strength and a scalar field (in Euclidean time) take the form
\begin{equation}
\label{Eq: Ansatz F}
    \bar{F}^I=F^I_{\tau y}(\theta_n)\,d\tau\wedge dy-i\,k_iF^I_{mi}(\theta_n)(y-1)\,d\tau\wedge d\theta_m+F^I_{mi}(\theta_n)d\theta_m\wedge d\varphi_i,
\end{equation}
\begin{equation}
    \label{Eq: background scalar}
    \bar{\phi}^A=\bar{\phi}^A(\theta_n),
\end{equation}
where $k_i$ is the same constant as in \eqref{Eq: BG ansatz}. The Einstein summation convention applies to the indices $m,i$.
\end{lemma}
\begin{proof}
To prove the above lemma, we follow the same argument as in Section 2.4 of \cite{Kunduri:2008rs} to higher dimensions. First, the existence of killing vectors $\tfrac{\partial}{\partial v},\tfrac{\partial}{\partial \varphi_i}$ eliminates $F_{vi}$ and $F_{ij}$. Second, the near-horizon limit \eqref{Eq: near-horizon limit} eliminates $F_{rm}$ and $F_{ri}$ and implies that $F_{vm}\propto r$. Third, the Maxwell equation implies relation between $F_{mv}$ and $F_{mi}$:
\begin{equation}
    F^{I}_{mv}(\theta_n)=F^{I}_{mi}(\theta_n)k^i(\theta_n)r.
\end{equation}
In the Gaussian null coordinates we find the gauge field strength takes the form 
\begin{equation}
\label{Eq: F ansatz Kunduri}
    F^I = F^I_{vr}(\theta_n)\, dv \wedge dr 
    + F^I_{mi}(\theta_n)\, d\theta_m \wedge \left(d\varphi_i + k^i(\theta_n) r\, dv \right),
\end{equation}
where, again, we sum over indices $m$ and $i$. Performing a coordinate transformations, one finds \eqref{Eq: Ansatz F}, see  Appendix \ref{App: Coordinate transformation} for more details.

Given the Killing vectors $\frac{\partial}{\partial v},\frac{\partial}{\partial\varphi_i}$, and the near-horizon limit, a scalar can only depend on $\theta_n$, so we have \eqref{Eq: background scalar}. 
\end{proof}

After proving \textit{Lemma \ref{lem: 2}}, we are ready to prove that \textit{Condition \ref{con: 1}} is not an independent assumption in \textit{Lemma \ref{lem: 1}}. Given the field strength, \eqref{Eq: F ansatz Kunduri}, and the metric, \eqref{Eq: g bar after Gamma}, direct inspection shows that the Einstein equation directly implies \textit{Condition \ref{con: 1}}. This proves that \textit{Condition \ref{con: 1}} applies to all black hole solutions of theory \eqref{Eq: Maxwell action} with symmetry as in \textit{Assumption \ref{ass: 1}}.


\subsection{The small temperature perturbation} 

We will show in Section \ref{Sec: Tensor zero modes} that  a group of tensor zero modes appears for any black hole that admits quadratic extremality. To understand the contribution of such zero modes to the gravitational path integral, one can regularize the corresponding IR divergence by introducing a small temperature perturbation to the background. Motivated by direct inspection of various backgrounds (see, for example,  \cite{Maulik:2024dwq}), here we impose the following empirical assumption about small temperature perturbation of the metric and matter fields:

\begin{assumption}
\label{ass: 2}
    The leading piece of the small temperature perturbation of the metric and matter fields of a near-horizon-extremal solution takes the form of
\begin{equation}
\label{Eq:g1+g2}
    \delta g_{\mu\nu}=T(\delta g_1+\delta g_2)_{\mu\nu},
\end{equation}
where
\begin{equation}
\label{Eq: delta g1}
\begin{aligned}
\delta g_{1\mu\nu}dx^\mu dx^\nu \;=\;&
\Bigl[\,y\bigl(1-y^{2}\bigr)\,
      \delta g_{\tau\tau}(\theta_n)\Bigr]\,
      d\tau^{2}\\[4pt]
&+ 2\Bigl[\,y(1+y)\,\delta g_{\tau i}^{(1)}(\theta_n)
           - i\,k_j\,y\,\delta g_{ij}(\theta_n)\Bigr]\,
      d\tau\,d\varphi_i\\[4pt]
&+ \Bigl[\,
      \frac{y\,\delta g_{yy}^{(1)}(\theta)}{1-y^{2}}
      \Bigr]\,dy^{2}
  \;+\; y\,\delta g_{mm}(\theta)\,d\theta_m^{2}
  \;+\; y\,\delta g_{ij}(\theta)\,d\varphi_id\varphi_j\;,
\end{aligned}
\end{equation}
\begin{equation}
\label{Eq: delta g2}
\begin{aligned}
\delta g_{2\mu\nu}dx^\mu dx^\nu \;=\;&
\Bigl[\,
      2\,i\,k_i\,(1-y)\bigl(
         y\,\delta g_{\tau i}^{(21)}(\theta_n)
         +    \delta g_{\tau i}^{(22)}(\theta_n)\bigr)
      \;-\;
      (1-y)\,\delta g_{yy}^{(2)}(\theta_n)
\Bigr]\,d\tau^{2}\\[4pt]
&+ 2\Bigl[\,
      y\,\delta g_{\tau i}^{(21)}(\theta_n)
      +     \delta g_{\tau i}^{(22)}(\theta_n)
\Bigr]\,d\tau\,d\varphi_i\\[4pt]
&+ \Bigl[\,
      \frac{\delta g_{yy}^{(2)}(\theta_n)}
           {(1-y)(1+y)^{2}}
\Bigr]\,dy^{2}\;,
\end{aligned}
\end{equation}
and
\begin{equation}
\label{Eq: Ansatz:deltaF}
    \delta F= \delta F_{\tau y}(y,\theta_n)\, d\tau\wedge dy\, +\delta F_{\tau m}(y,\theta_n)\,d\tau\wedge d\theta_m +\delta F_{yi}(y,\theta_n)\, dy\wedge d\varphi_i+\delta F_{mi}(y,\theta_n)\, d\theta_m \wedge d\varphi_i,
\end{equation}
\begin{equation}
\label{Eq: Ansatz: dleta phi}
\delta\phi^A=\delta\phi^A(y,\theta)
\end{equation}
where, again, the index $m$ in \eqref{Eq: delta g1} is only summed once. All the undetermined functions in \eqref{Eq: delta g1}-\eqref{Eq: Ansatz: dleta phi} are assumed to be smooth functions in $\theta_n$. The upper indices in the parentheses are just used to differentiate them. Every function corresponds to a mode that drives the black hole away from extremality.
\end{assumption}
\begin{remark}
The above assumption is a structural Ansatz supported by all examples we know in $D\ge4$ dimensions. In four and five dimensions, the above assumption agrees with near-extremal Kerr-Newman-AdS$_4$ and Kerr-AdS$_5$ black holes in \cite{Maulik:2024dwq} and near-cold Kerr-dS$_4$ black hole in Section \ref{Sec: Kerr_dS}. We also checked that, in six dimensions, the perturbed metric above agrees with that of a six-dimensional Myers-Perry black hole. 
\end{remark}

\begin{remark}
We have separated the order $T$ metric into two parts according to different asymptotic behavior of their wavefunction density
\begin{equation}
    ||\delta g^2||_{den}=\sqrt{\bar{g}}\,\delta g^*_{\mu\nu}\delta g^{\mu\nu},
\end{equation}
where $\bar{g}$ is the extremal background. At the $AdS_2$ boundary, $y\to\infty$, $||\delta g_1^2||_{den}\to \mathcal{O}(y^2)$, while $||\delta g_2^2||_{den}$ is finite. As mentioned in Section \ref{Sec:A brief review to the procedure}, the zero modes reflect the asymptotic symmetry of the near-horizon-extremal geometry. Their eigenvalues are lifted by modes in $\delta g_1$. As we will see in Section \ref{Sec: cancellation}, modes in $\delta g_2$ do not contribute to the lifted eigenvalues of the tensor zero modes or contribute only  to the $n=2$ lifted eigenvalue, which will lead to an overall constant coefficient of the partition function. In \cite{Maulik:2024dwq}, it is mentioned that the $\tfrac{3}{2}\log{T}$ correction is independent of the choice of ensemble. We will see in Section \ref{Sec: Kerr_dS}, in the example of Kerr-dS black holes, that the more precise statement is that the ensemble choice can only influence $\delta g_2$.
\end{remark}

For convenience, we give the following definitions before we start to investigate the quantum corrections contributed from the tensor zero modes.

\begin{definition}
\label{def: 2}
    A \textit{black hole one-parameter family} is defined as a map $\mathcal{E}$
\begin{equation}
    \mathcal{E}:[0,\infty)\to \mathcal{S},\quad T\to\mathcal{E}(T),
\end{equation}
where $\mathcal{S}$ is the space of black hole solutions of the theory given by \eqref{Eq: Maxwell action}, and $T$ is the temperature of the black hole. 
\end{definition}

\begin{definition}
    \label{def: 3}
    We define a \textit{black hole one-parameter family} of \textit{Type A} as a \textit{black hole one-parameter family} in  the theory defined by the Lagrangian given in \eqref{Eq: Maxwell action} that satisfies: (1) The near-horizon limit of $\mathcal{E}(0)$ takes the form of \eqref{Eq: BG ansatz} and \eqref{Eq: Ansatz F}; (2)\textit{Assumption \ref{ass: 2}}, {\it  i.e.}, as $T\to0$, the leading order of $\mathcal{E}(T)-\mathcal{E}(0)$ is given by \eqref{Eq:g1+g2}-\eqref{Eq: Ansatz: dleta phi} in the near-horizon region.
\end{definition}

\section{Tensor zero modes}
\label{Sec: Tensor zero modes}
\subsection{Existence of tensor zero modes}
Having introduced the precise definitions pertaining near-horizon-extremal configurations, we are now prepared to describe the tensor zero modes. These zero modes are generated by a set of large gauge transformations of the metric field. They reflect the asymptotic symmetry of \eqref{Eq: BG ansatz 4d}, and are annihilated by the Lichnerowicz operator \eqref{Eq: Linchnerowichz operator}. We give the following lemma for the explicit structure of the tensor zero modes:

\begin{lemma}
    \label{lem: 3}
    For a near-extremal black hole one-parameter family of \textit{Type A}, the Lichnerowicz  operator  \eqref{Eq: Linchnerowichz operator} admits a set of tensor zero modes at $T=0$ of the following form

\begin{equation}
\label{Eq: Tensor Zero Modes}
    h^{(n)}_{\mu\nu}dx^\mu dx^\nu=C_n\,g_1(\theta_m)e^{i n\tau}\left(\frac{y-1}{y+1}\right)^{\frac{|n|}{2}}\left(-d\tau^2+\frac{dy^2}{(y^2-1)^2}+2i\frac{n}{|n|}\frac{d\tau dy}{y^2-1}\right),
\end{equation}
where $|n|\ge2$ is an integer, $C_n$ are normalization constants and $g_1(\theta)$ is the same function as in \eqref{Eq: BG ansatz 4d}.
\end{lemma}

\begin{proof}
    We prove that $h^{(n)}_{\mu\nu}$ is indeed a zero mode in two steps. We first prove that $h^{(n)}_{\mu\nu}$ sits in the ``traceless and divergence-free" gauge \eqref{Eq: gauge}, which annihilates the gauge-fixing term \eqref{Eq: gauge fixing}. Then we prove that $h^{(n)}_{\mu\nu}$ can be generated by a (large) diffeomorphism, which annihilates the Einstein-Hilbert parts. By these two steps, $h^{(n)}_{\mu\nu}$ will automatically be annihilated by the Lichnerowicz operator. The first step follows from direct evaluation 
\begin{equation}
\label{Eq: metric gauge condition}
    \bar{h}^\mu_{\,\,\mu}= h_{\mu\nu}\bar{g}^{\mu\nu}=0,\quad\bar{\nabla}^\mu h_{\mu\nu}=0,
\end{equation}
where $\bar{\nabla}$ is the covariant derivative operator associated with the extremal metric $\bar{g}_{\mu\nu}$. For the second step, we need to find a vector field that generates \eqref{Eq: Tensor Zero Modes}. If the black hole is static, the tensor zero modes are generated by a vector field that has components only on $AdS_2$. In this case, one can convert the above condition \eqref{Eq: metric gauge condition} to a second order PDE of a scalar field. More details can be found in Section 3.2 and Equation (4.9) of \cite{Blacker:2025zca}. For rotating black holes, as the background metric has $\tau\varphi_i$ components, we need to include the $\varphi_i$ components of the vector field:
\begin{equation}
\label{eq: Diffeo Vector}
    \xi^{(n)\mu}\frac{\partial}{\partial x^\mu}=e^{in\tau}\left(f_2(y)\frac{\partial}{\partial\tau}+f_1(y)\frac{\partial}{\partial y}+f_3^{(j)}(y)\frac{\partial}{\partial\varphi_j}\right),
\end{equation}
Assuming the same functions $f_1,f_2$ as the static case and demanding $h_{\mu\nu}=\mathcal{L}_{\xi}\bar{g}_{\mu\nu}$ only have non-zero components on $AdS_2$ gives
\begin{align}
    f_1&=-in(n+y)\left(\frac{y-1}{y+1}\right)^{\frac{|n|}{2}},\\
    f_2&=\frac{if_1'(y)}{n},\\
    f_3^{(j)}&=\frac{k_j\left(y-1\right)f_1'(y)-k_jf_1(y)}{|n|},
\end{align}
where $k_j$ are the constants in \eqref{Eq: BG ansatz}. The modes with $n=0,\pm1$ are excluded because they correspond to the isometries of the near-horizon-extremal geometry $\bar{g}$ and the Lie derivative vanishes. One can verify that this is indeed the vector field we are looking for:
\begin{equation}
h^{(n)}_{\mu\nu}=\mathcal{L}_{\xi^{(n)}}\bar{g}_{\mu\nu}.
\end{equation}
Now we have shown that the set of zero modes given in \eqref{Eq: Tensor Zero Modes} is a diffeomorphism which has no contribution to the Einstein-Hilbert action \eqref{Eq: Maxwell action}. At the same time, the zero modes \eqref{Eq: Tensor Zero Modes} remain in the ``traceless and divergence-free" gauge, so they do not change the action through the gauge-fixing term \eqref{Eq: gauge fixing}. This finishes the proof that they are zero modes of the Lichnerowicz operator.  
\end{proof}

\begin{remark}
    Note that the zero modes \eqref{Eq: Tensor Zero Modes} are normalizable while the generating vector field is non-normalizable. Indeed, one can verify that the norms are 

\begin{eqnarray}
    ||h_{\mu\nu}||^2&=& \int d^Dx\,\sqrt{-g}\,h^*_{\mu\nu}h^{\mu\nu}  <\infty, \nonumber \\
    ||\xi_\mu||^2&=& \int d^Dx\,\sqrt{-g}\,\xi^{*}_{\mu}\xi^{\mu}\to \infty.
\end{eqnarray}
This is precisely the property we expect for large diffeomorphisms. We note that our construction is universal, depends only $g_1(\theta_m)$. It recovers the spherically symmetric examples \cite{Sen:2012kpz,Iliesiu:2022onk,Charles:2015eha} and, more importantly, the  rotating cases \cite{Kapec:2023ruw,Rakic:2023vhv,Maulik:2025phe,Blacker:2025zca}
\end{remark}

\subsection{Relation to the boundary curve reparameterization}
We have thus far treated the zero modes from a higher-dimensional ($D\ge4$)  perspective, it is quite instructive to connect with the two-dimensional point of view which constitutes the original conceptual source.  This subsection explicitly shows  how the diffeomorphisms in \eqref{eq: Diffeo Vector} relate to the reparameterization modes of the boundary curve of $AdS_2$ in the framework of JT gravity.
Let us first review the boundary curve reparameterization modes in JT gravity. Following \cite{Mertens:2022irh}, the reparameterization modes are defined as follows. Start from Poincar\'e coordinates of Euclidean $AdS_2$
\begin{equation}
    ds^2=\frac{dT^2+dZ^2}{Z^2}.
\end{equation}
As a boundary condition, we fix the induced metric $h_{JT}$ on the boundary curve to satisfy $\sqrt{h_{JT}}=\tfrac{1}{\epsilon}$, where $\epsilon$ is infinitesimal. We also fix the dilaton field asymptotics to be $\Phi_{JT}|_\partial=\tfrac{a}{2\epsilon}$, where $a$ is a constant. Under these conditions, the boundary curve can be written as $(T,Z)=(F(\tilde{\tau}),\epsilon F'(\tilde{\tau}))$.  The JT action  reduces to its Gibbons-Hawking extrinsic curvature contribution which becomes the standard Schwarzian action
\begin{equation}
\label{Eq: Schwarzian action}
    I_{JT}=-C\int d\tilde{\tau}\{F,\tilde{\tau}\},\qquad C=\frac{a}{16\pi G_N},
\end{equation}
where $\{\dots\}$ denotes the Schwarzian derivative,
\begin{equation}
    \{F,\tilde{\tau}\}\equiv \frac{F'''}{F'}-\frac{3}{2}\left(\frac{F''}{F'}\right)^2
\end{equation}
The saddle point that solves the equation of motion derived from the above Schwarzian action \eqref{Eq: Schwarzian action} is given by 
\begin{equation}
    F(\tilde{\tau})=\tan(\frac{\pi}{\beta}\tilde{\tau}), \qquad Z(\tilde{\tau})=\epsilon F'(\tilde{\tau})=\epsilon\frac{\pi}{\beta}\sec^2(\frac{\pi}{\beta}\tilde{\tau}).
\end{equation}
As we are mostly interested in the way the $AdS_2$ geometry appears in the near-horizon region of near-extremal black holes, it is more convenient to use another coordinate system, the so-called black hole patch,
\begin{equation}
\label{Eq: black hole patch}
    ds^2=\frac{4\pi^2}{\beta^2}\sinh^2{\rho}\,d\tilde{\tau}^2+d\rho^2
\end{equation}
The coordinate transformation from the Poincar\'e patch to the black hole patch is
\begin{equation}
    Z(\tau,\rho)=
\frac{1}{\cos(\kappa\tau)\,\sinh\rho+\cosh\rho},\quad    T(\tau,\rho)=
\frac{\sin(\kappa\tau)\,\sinh\rho}{\cos(\kappa\tau)\,\sinh\rho+\cosh\rho},
\end{equation}
with $\kappa=\frac{2\pi}{\beta}$. The boundary curve now lives at large $\rho$. At $\rho\to\infty$ limit
\begin{equation}
    Z(\tau,\rho)
=\frac{e^{-\rho}}{\cos^2\!\big(\frac{\kappa\tau}{2}\big)},\quad T(\tau,\rho)
=
\tan\!\Big(\frac{\kappa\tau}{2}\Big).
\end{equation}
In the black hole patch, the saddle is mapped as $(\tilde{\tau},\rho)=(\tilde{\tau},-\ln{(\epsilon\tfrac{\pi}{\beta})})$. Fluctuations around this saddle can be described by a function $f\in \mathrm{Diff} \,S^1$ via $(\tilde{\tau},\rho)=(\tilde{\tau}+f(\tilde{\tau}),-\ln{(\epsilon\tfrac{\pi}{\beta})-f'(\tilde{\tau})})$. Given the periodicity, one can expand this function, $f(\tilde{\tau})$, in terms of its Fourier modes to arrive at the  so-called boundary curve reparameterization modes, $\epsilon_n$,
\begin{equation}
\label{Eq: Boudary Reparameterization}
    f(\tilde{\tau})=\frac{\beta}{2\pi}\sum_{n\ge2}e^{-\frac{2\pi}{\beta}in\tilde{\tau}}\epsilon_n+\mathrm{h.c.},
\end{equation}
where we have excluded the unphysical $n=-1,0,1$ modes, which reflect the isometries of $AdS_2$.
To compare the above reparameterization modes with the diffeomorphism \eqref{eq: Diffeo Vector}, we first perform a coordinate transformation
\begin{equation}
    \tau\to \frac{2\pi}{\beta}\tilde{\tau},\quad y\to\cosh{\rho}.
\end{equation}
The $AdS_2 $ part of the near-horizon extremal geometry \eqref{Eq: BG ansatz} becomes the same as the black hole patch \eqref{Eq: black hole patch}. Under this coordinate transformation and $\rho \to \infty$, the vector field \eqref{eq: Diffeo Vector} becomes
\begin{equation}
     \xi^{(-n)\mu}\frac{\partial}{\partial x^\mu}=e^{-\frac{2\pi}{\beta}in\tilde{\tau}}\left(\frac{\beta}{2\pi}\frac{\partial}{\partial\tilde{\tau}}+in\frac{\partial}{\partial \rho}+ik_i\frac{\partial}{\partial\varphi_i}\right).
\end{equation}
As we are interested in comparing it with the $AdS_2$ reparameterization modes, we project it to the $(\tilde{\tau},\rho)$-plane on which its action on the  boundary curve is
\begin{equation}
    (\tilde{\tau},\rho)=(\tilde{\tau}+\xi^{(-n)\tilde{\tau}}\epsilon_n,-\ln{(\epsilon\frac{\pi}{\beta})}+\xi^{(-n)\rho}\epsilon_n).
\end{equation}
This is the same as the boundary curve reparameterization modes \eqref{Eq: Boudary Reparameterization}. We now see a one-to-one correspondence between $\xi^{(-n)}$ and $\epsilon_n$. One can redefine the index $n$ to have a cleaner dictionary between these two languages of the Schwarzian modes:
\begin{equation}
    \xi^{(n)}\leftrightarrow \epsilon_n.
\end{equation}

The above map illustrates the physics of quantum JT gravity without recurring to dimensional reduction of the higher dimensional theory to two dimensions.

\section{The cancellation and lifted eigenvalues}
\label{Sec: cancellation}
In this section, we prove that, after being regularized by the small temperature perturbation \eqref{Eq:g1+g2} and \eqref{Eq: Ansatz:deltaF}, the lifted eigenvalues of the tensor zero modes \eqref{Eq: Tensor Zero Modes} contribute a $\frac{3}{2}\log{T}$ correction to the thermodynamic entropy. We first investigate the role of the matter sector in this calculation.

\subsection{The cancellation}
Instead of substituting the small temperature perturbations \eqref{Eq:g1+g2} and \eqref{Eq: Ansatz:deltaF} into the full Lichnerowicz operator \eqref{Eq: Linchnerowichz operator}, we first simplify the calculation by proving an empirical observation made in \cite{Maulik:2024dwq}: only two terms in the Lichnerowicz operator contribute to the lifted eigenvalues.

We start by first grouping all terms in the Lichnerowicz operator \eqref{Eq: Linchnerowichz operator} as
\begin{equation}
\begin{alignedat}{2}
{\rm Term\, 1a:}&\, h_{\alpha\beta}^*\ \delta(\tfrac{1}{2}g^{\alpha\mu}g^{\beta\nu}\square)h_{\mu\nu}
&\hspace{0.8em}{\rm Term\, 1b:}&\, h_{\alpha\beta}^*\ \delta(-\tfrac{1}{4}g^{\alpha\beta}g^{\mu\nu}\square )h_{\mu\nu} \\
{\rm Term\, 2:}&\, h_{\alpha\beta}^*\ \delta(R^{\alpha\mu\beta\nu})h_{\mu\nu} && \\
{\rm Term\, 3a:}&\, h_{\alpha\beta}^*\ \delta(R^{\alpha\mu}g^{\beta\nu})h_{\mu\nu}
&\hspace{0.8em}{\rm Term\, 3b:}&\, h_{\alpha\beta}^*\ \delta(-R^{\alpha\beta}g^{\mu\nu})h_{\mu\nu} \\
{\rm Term\, 4a:}&\, h_{\alpha\beta}^*\ \delta(-\tfrac{1}{2}R g^{\alpha \mu}g^{\beta\nu} )h_{\mu\nu}
&\hspace{0.8em}{\rm Term\, 4b:}&\, h_{\alpha\beta}^*\ \delta(\tfrac{1}{4}R g^{\alpha\beta}g^{\mu\nu})h_{\mu\nu} \\
{\rm Term\, 5a:}&\, h_{\alpha\beta}^*\ \delta(\tfrac{1}{2}V\, g^{\alpha \mu}g^{\beta\nu})h_{\mu\nu}
&\hspace{0.8em}{\rm Term\, 5b:}&\, h_{\alpha\beta}^*\ \delta(-\tfrac{1}{4}V\, g^{\alpha \beta}g^{\mu\nu})h_{\mu\nu} \\
{\rm Term\, 6a:}&\, h_{\alpha\beta}^*\ \delta(\tfrac{1}{2}\,\mathcal{C}_{IJ} F^I_{\rho\sigma}F^{J\rho\sigma} g^{\alpha \mu} g^{\beta \nu})h_{\mu\nu}
&\hspace{0.8em}{\rm Term\, 6b:}&\, h_{\alpha\beta}^*\ \delta(-\tfrac{1}{4} \,\mathcal{C}_{IJ}F^I_{\rho\sigma}F^{J\rho\sigma}g^{\alpha \beta} g^{\mu \nu})h_{\mu\nu} \\
{\rm Term\, 7:}&\, h_{\alpha\beta}^*\ \delta(-2\,\mathcal{C}_{IJ}F^{I\alpha \mu} F^{J\beta \nu})h_{\mu\nu} && \\
{\rm Term\, 8a:}&\, h_{\alpha\beta}^*\ \delta(-4\,\mathcal{C}_{IJ} F^{I\alpha \gamma} F^{J\mu}{}_{ \gamma} g^{\beta \nu})h_{\mu\nu}
&\hspace{0.8em}{\rm Term\, 8b:}&\, h_{\alpha\beta}^*\ \delta(2\,\mathcal{C}_{IJ}F^{I\alpha \gamma} F^{J\beta}{}_{\gamma} g^{\mu \nu})h_{\mu\nu} \\
{\rm Term\, 9a:}&\, h_{\alpha\beta}^*\ \delta(\tfrac{1}{4}\,f_{AB}\partial_\rho\phi^A\partial^{\rho}\phi^B\, g^{\alpha\mu}g^{\beta\nu})h_{\mu\nu}
&\hspace{0.8em}{\rm Term\, 9b:}&\, h_{\alpha\beta}^*\ \delta(-\tfrac{f_{AB}}{8}\partial_\rho\phi^A\partial^{\rho}\phi^B\, g^{\alpha\beta}g^{\mu\nu})h_{\mu\nu} \\
{\rm Term\, 10a:}&\, h_{\alpha\beta}^*\ \delta(-\,f_{AB}\partial^\alpha\phi^A\partial^{\mu}\phi^B\, g^{\beta\nu})h_{\mu\nu}
&\hspace{0.8em}{\rm Term\, 10b:}&\, h_{\alpha\beta}^*\ \delta(\tfrac{1}{2}\,f_{AB}\partial^\alpha\phi^A\partial^{\beta}\phi^B\, g^{\mu\nu})h_{\mu\nu}
\end{alignedat}
\label{eq:cancelling terms KN-AdS4}
\end{equation}
where $\delta$ indicates that the above expressions are of order $T$. As $h_{\mu\nu}$ are zero modes, the $O(1)$ result is zero. There were two main findings in \cite{Maulik:2024dwq} regarding the structure of the terms:
\begin{itemize}
    \item In \eqref{eq:cancelling terms KN-AdS4}, all terms with the trace of the zero modes, {\it i.e.}, the ``b" terms, vanish by themselves.
    \item ${\rm Term \, 3a}+{\rm Term\, 4a}+{\rm Term \, 5a}+{\rm Term \, 6a}+{\rm Term \, 7}+{\rm Term \, 8a}+{\rm Term \, 9a}+{\rm Term \, 10a}=0.$ 
\end{itemize}
Therefore, only two terms contribute non-trivially. We summarize this observation via the following lemma
\begin{lemma}
\label{lem: 4}
For a near-extremal black hole one-parameter family of \textit{type A}, the action of the Lichnerowicz operator \eqref{Eq: Linchnerowichz operator} on the tensor zero modes given by \textit{Lemma \ref{lem: 3}}, reduces to 
\be
\label{eq:Cancellation}
h_{\alpha\beta}^*\Delta^{\alpha \beta, \mu \nu}_{L}h_{\mu\nu}={\rm Term \, 1a}+{\rm Term\, 2}= h_{\alpha\beta}^*\ \delta(\frac{1}{2}g^{\alpha\mu}g^{\beta\nu}\square+R^{\alpha\mu\beta\nu})h_{\mu\nu}.
\ee
\end{lemma}

\begin{proof}
    First, we consider the ``b" terms. In Section \ref{Sec: Tensor zero modes}, it was shown that the zero modes $h_{\mu\nu}$ are traceless with respect to the background metric $\bar{g}^{\mu\nu}$. Therefore, the trace of $h_{\mu\nu}$ is at least  linear in $T$,
\begin{equation}
h^{\mu}_{\,\,\mu}=h_{\mu\nu}g^{\mu\nu}=\mathcal{O}(T).
\end{equation}
This indicates that all terms with $(h^{\mu}_{\,\,\mu})^2$ vanish at order $T$. This annihilates Term 1b, Term 4b, Term 5b, Term 6b, and Term 9b. There are three ``b" terms left: Term 3b, Term 8b, and Term 10b, as they are linear in $h^{\mu}_{\,\,\mu}$. Applying our background metric \eqref{Eq: BG ansatz} and the tensor zero modes \eqref{Eq: Tensor Zero Modes}, direct evaluation shows
\begin{equation}
\label{Eq: Ricci trace} h_{\alpha\beta}\bar{R}^{\alpha\beta}=0.
\end{equation}
Therefore, 
\begin{equation}
    h_{\alpha\beta}R^{\alpha\beta}=\mathcal{O}(T).
\end{equation}
We conclude that Term 3b also vanishes at order $T$. To show that Term 8b and Term 10b vanish, we start from the Einstein equation
\begin{equation}
\label{Eq: Einstein equation}
\begin{aligned}
    R_{\mu\nu} - \tfrac{1}{2} g_{\mu\nu} R
&= \frac{1}{2}f_{AB}(\phi)\,\partial_\mu \phi^{A}\partial_\nu \phi^{B}-\frac{1}{4}\,g_{\mu\nu}\,f_{AB}(\phi)\,\partial_\rho \phi^{A}\partial^\rho \phi^{B} \\
&\quad - \frac{1}{2}g_{\mu\nu}\,V(\phi) + 2\,\mathcal{C}_{IJ}(\phi)\left( 
      F^{I}{}_{\mu\rho}F^{J}{}_{\nu}{}^{\rho}
      - \tfrac{1}{4}\,g_{\mu\nu}\,F^{I}_{\rho\sigma}F^{J\,\rho\sigma}
   \right).
\end{aligned}
\end{equation}
Applying the above Einstein equation to the extremal configuration and contracting both sides with $h^{\mu\nu}$ gives
\begin{equation}
\label{Eq: FF trace}
\begin{aligned}
h^{\mu\nu} \,\bar{R}_{\mu\nu}
&= \frac{1}{2}h^{\mu\nu}\, f_{AB}(\bar{\phi})\, \partial_\mu \bar{\phi}^{A}\, \partial_\nu \bar{\phi}^{B}
  + 2\,\mathcal{C}_{IJ}(\bar{\phi})\, h^{\mu\nu}\, \bar{F}^{I}{}_{\mu\rho}\, \bar{F}^{J}{}_{\nu}{}^{\rho}.
\end{aligned}
\end{equation}
As $\bar\phi^A$ depends only on $\theta_m$  and , at the same time, $h^{mm}$ vanishes, we have 
\begin{equation}
    h^{\mu\nu}\, f_{AB}(\bar{\phi})\, \partial_\mu \bar{\phi}^{A}\, \partial_\nu \bar{\phi}^{B}=0,
\end{equation}
which implies that Term 10b vanishes at order $T$.
Combining \eqref{Eq: Ricci trace} and \eqref{Eq: FF trace}, we find that Term 8b also vanishes at order $T$. 
The first point of the empirical finding has been proved: Up to order $T$, all ``b" terms vanish 
by themselves. 

For the other terms, we use the Einstein equation again. With a slight modification of the indices,
\begin{align}
\label{Eq: Modified EE}
\nonumber
& R^{\alpha\mu}\,g^{\beta\nu}
-\tfrac{1}{2}\,g^{\alpha\mu}\,R\,g^{\beta\nu}
+\tfrac{1}{2}\, g^{\alpha\mu}\,V\,g^{\beta\nu}
+\tfrac{1}{2}\,\mathcal{C}_{IJ}\,\,g^{\alpha\mu}\,F^{I}_{\rho\sigma}F^{J\,\rho\sigma}g^{\beta\nu}-2\,\mathcal{C}_{IJ}F^{I\,\alpha\rho}\,F^{J\mu}{}_{\rho}g^{\beta\nu}
\\
&\qquad
+\tfrac{1}{4}\,g^{\alpha\mu}\,f_{AB}\,\partial_{\rho}\phi^{A}\partial^{\rho}\phi^{B}\,g^{\beta\nu} 
- \tfrac{1}{2}f_{AB}\,\partial^{\alpha}\phi^{A}\partial^{\mu}\phi^{B}\,g^{\beta\nu}
= 0 .
\end{align}
After contracting twice with $h_{\mu\nu}$, we find \eqref{Eq: Modified EE} has a very similar form to the sum of the rest terms, 
\be
\label{Eq: term a s}
\begin{aligned}
&{\rm Term \, 3a}+{\rm Term\, 4a}+{\rm Term \, 5a}+{\rm Term \, 6a}+{\rm Term \, 7}+{\rm Term \, 8a}+{\rm Term \, 9a}+{\rm Term \, 10a}\\
=&h_{\alpha\beta}^*\Bigl(R^{\alpha\mu}\,g^{\beta\nu}
-\tfrac{1}{2}\,g^{\alpha\mu}\,R\,g^{\beta\nu}
+\tfrac{1}{2}\, g^{\alpha\mu}\,V\,g^{\beta\nu}
+\tfrac{1}{2}\,\mathcal{C}_{IJ}\,g^{\alpha\mu}\,F^{I}_{\rho\sigma}F^{J\,\rho\sigma}\, g^{\beta\nu}\underline{-2\,\mathcal{C}_{IJ}\,F^{I\alpha\mu}\,F^{J\beta\nu}}
\\
&\qquad\underline{-4\,\mathcal{C}_{IJ}F^{I\alpha\rho}F^{J\mu}{}_{\rho} g^{\beta\nu}}\,
+\tfrac{1}{4}\,g^{\alpha\mu}\,f_{AB}\,\partial_{\rho}\phi^{A}\partial^{\rho}\phi^{B}\,g^{\beta\nu}\underline{-f_{AB}\,\partial^{\alpha}\phi^{A}\partial^{\mu}\phi^{B}\,g^{\beta\nu}}
\Bigr)h_{\mu\nu},
\end{aligned}
\ee
with difference underlined. Now we substitute the small temperature configuration, $\{g=\bar{g}+\delta g,\phi=\bar\phi+\delta\phi,F=\bar{F}+\delta F\}$. Direct evaluation shows 
\begin{equation}
h_{\alpha\beta}^*\,\mathcal{C}_{IJ}\left(-2F^{I\alpha\mu}F^{J\beta\nu}-2g^{\beta\nu}F^{I\alpha\rho}F^{J\mu}{}_{\rho}\right)h_{\mu\nu}=0,
\end{equation}
and
\begin{equation}
h_{\alpha\beta}^*f_{AB}\,\partial^{\alpha}\phi^{A}\partial^{\mu}\phi^{B}\,g^{\beta\nu}h_{\mu\nu}=0,
\end{equation}
which fills the gap between \eqref{Eq: Modified EE} and \eqref{Eq: term a s}. The second point of the empirical findings is then proved.
\end{proof}

\subsection{The lifted eigenvalues}
\label{Subsec: lifted eigenvalues}
\textit{Lemma \ref{lem: 3}} shows the existence of the tensor zero modes. With the help of \textit{Lemma \ref{lem: 4}}, the Lichnerowicz operator is simplified to \eqref{eq:Cancellation}. We are now ready to prove a theorem stating the universality of the $\tfrac{3}{2}\log{T}$ correction. 

\begin{theorem}
\label{the: 1}
A black hole one-parameter family of \textit{type A} in dimensions $D=4,5,6$ admits a set of tensor zero modes at extremality, $T\to0$. With a mild assumption, \textit{Assumption \ref{ass: 3}}, shown below, the contribution of these zero modes to the thermodynamic entropy at low temperatures, computed via the Euclidean path integral,  always contains a temperature dependence of the form $\frac{3}{2}\log{T}$.
\end{theorem}

\begin{proof}[Proof (in four dimension)]
    We first add the linear-in-temperature configuration to the background
\begin{equation}
    g=\bar{g}+\delta g, \quad \Phi=\bar\Phi +\delta \Phi.
\end{equation} 
Then keeping the order $T$ contribution of \eqref{eq:Cancellation} and integrating over the whole spacetime gives the lifted eigenvalues. This calculation is cumbersome but direct. After determining the normalization of the zero modes, the lifted eigenvalues are given by
\begin{equation}
\label{Eq: the intigral}
    \Lambda^{(n)}=\int d^4x\; h_{\alpha\beta}^*\Delta^{\alpha \beta, \mu \nu}_{L}h_{\mu\nu}=\int d^4x\; \lambda(\theta,y),
\end{equation}
where the integral in $\tau,\,\varphi$ directions is simple. The integral in the $y$ direction can be done explicitly,
\begin{equation}
    \Lambda^{(n)}=\int d\theta\; \tilde{\lambda}(\theta) \,n\,T,\quad n>2,
\end{equation}
where $\tilde{\lambda}(\theta)$ depends on the functions in \eqref{Eq: BG ansatz 4d} and \eqref{Eq:g1+g2}, and it has no $n$ dependence. For $n=2$, the eigenvalue does not fall into the above pattern
\begin{equation}
    \Lambda^{(2)}=\int d\theta\; \tilde{\lambda}_2(\theta) \,T.
\end{equation}
For $n=0,1$, the eigenvalues are not lifted due to the $SL(2,R)$ symmetry of the extremal background.
We present the explicit expressions of $\lambda(\theta,y)$, $\tilde{\lambda}(\theta)$, and $\tilde{\lambda}_2(\theta)$ in Appendix \ref{App: Lifted eigenvalue}.
We assume
\begin{assumption}
    \label{ass: 3} The expressions for 
    $\tilde{\lambda}(\theta)$ and $\tilde{\lambda}_{2}(\theta)$ are assumed to be smooth in the interval $\theta\in[0,\pi]$, and
\begin{equation}
    \tilde{\Lambda}=\int_0^\pi\tilde{\lambda}(\theta)d\theta\neq0,\qquad \tilde{\Lambda}_2=\int_0^\pi\tilde{\lambda}_2(\theta)d\theta\neq0. 
\end{equation}
\end{assumption}

Then we have
\begin{equation}
\begin{aligned}
\Lambda^{(n)}=&\tilde{\Lambda}\;n\;T,\quad n\ge3,\\
\Lambda^{(2)}=&\tilde{\Lambda}_2\;T, \quad n=2.
\end{aligned}
\end{equation}
The one-loop correction to the partition function is given by
\begin{equation}
    \delta Z=(\det\Delta^{\alpha \beta, \mu \nu}_{L})^{-\frac{1}{2}}=(\tilde{\Lambda}_2\;T)^{-1}\prod_{n\ge3}(\tilde{\Lambda}\;n\;T)^{-1}
\end{equation}
Using zeta function regularization, we can evaluate the infinite product as
\begin{equation}
    \prod_{n\ge 3} \frac{\alpha}{n}
   \;=\;
   \frac{2}{\sqrt{2\pi}\,\alpha^{5/2}}
\end{equation}
The final answer takes the form
\begin{equation}
\label{Eq: theorem: 3/2logT}
    \delta \log{Z}=\log{\left(\frac{2\tilde{\Lambda}^{5/2}}{\sqrt{2\pi}\,\tilde{\Lambda}_2}T^{3/2}\right)}\sim\frac{3}{2}\log{T}+\log{\left(\frac{2\tilde{\Lambda}^{5/2}}{\sqrt{2\pi}\,\tilde{\Lambda}_2}\right)}.
\end{equation}
This completes the proof of \textit{Theorem \ref{the: 1}} in four dimensions.
\end{proof} 

\begin{remark}
    The $\frac{3}{2}\log{T}$ dependence comes from the one-loop Gaussian integral. When the lifted eigenvalues of the tensor zero modes are negative (see, for example,  the Nariai limit treated in \cite{Maulik:2025phe,Blacker:2025zca}), the logarithm develops an imaginary part. However, in this case, the Gaussian integral is not convergent. Our result here should be understood as an analytical continuation, which is not necessarily  physical. 
\end{remark}

\begin{remark}
We also carry out the proof for dimensions $D=5,6$. As the calculation is very similar to $D=4$ and the explicit intermediate results are lengthy, we only show the procedure and the explicit intermediate steps in four dimensions above. Moreover, we predict no conceptual bottleneck for $D\geq7$. The reason why we only prove the theorem up to six dimensions is that this proof depends on  explicit tensor calculations that are computer aided. Performing the same calculations in higher dimensions requires more computational resources and is unlikely to reveal deeper physical understanding on this subject.
\end{remark}

\begin{remark}
    Given that  $U(1)^{\lfloor\frac{D-1}{2}\rfloor}$ is a subgroup of $SO(D-1)$, spherically symmetric black holes are naturally included in \textit{Theorem \ref{the: 1}.}
\end{remark}

\section{Planar symmetric black branes}
\label{Sec: Planar}
The above lemmas and the theorem focus on black holes with axial or spherical symmetry. In this section, we consider a simple generalization to planar-symmetric backgrounds. The planar-symmetric case is simpler than the axially symmetric one, but it is quite relevant in applications of the AdS/CFT correspondence to certain low-temperature aspects of condensed matter (see, for example, \cite{Liu:2024gxr,PandoZayas:2025snm,Cremonini:2025yqe} ). We, therefore,  find a discussion of the planar-symmetric case pertinent.  This section does not serve as a proof. We only provide the central ingredients for the case of black branes.

We stay in four dimensions and keep only one gauge field as the matter sector for simplicity. The analogy of \eqref{Eq: BG ansatz}, in coordinates $\{\tau,y,x_1,x_2\}$, is given by
\begin{equation}    ds^2=g_1\!\left[(y^{2}-1)\,d\tau^{2}+\frac{dy^{2}}{y^{2}-1}\right]\;+\;g_2\left(dx_1^{2}+\,dx_2^2\right),
\end{equation}
where $g_1,\,g_2$ are constants. The background gauge field strength, compared to \eqref{Eq: Ansatz F}, is given by 
\begin{equation}
    F=F_{\tau y}\,d\tau \wedge dy,
\end{equation}
where $F_{\tau y}$ is a constant. The tensor zero modes still take the form of \eqref{Eq: Tensor Zero Modes} with $g_1(\theta)$ replaced by a constant $g_1$.
The assumption for the small temperature perturbation is given by 
\begin{equation}
\label{Eq:g1+g2 planar}
    \delta g_{\mu\nu}=T(\delta g_1+\delta g_2)_{\mu\nu},
\end{equation}
where
\begin{equation}
\label{Eq: delta g1 xy}
\begin{aligned}
\delta g_{1\mu\nu}dx^\mu dx^\nu \;=\;&
\,y\bigl(1-y^{2}\bigr)\,
      \delta g_{\tau\tau}\,
      d\tau^{2}+ \,
      \frac{y\,\delta g_{yy}^{(1)}}{1-y^{2}}
      \,dy^{2}
  \;+\; y\,\delta g_{x_1 x_1}\,dx_1^{2}
  \;+\; y\,\delta g_{x_2 x_2}\,dx_2^{2}\;,
\end{aligned}
\end{equation}

\begin{equation}
\label{Eq: delta g2 xy}
\begin{aligned}
\delta g_{2\mu\nu}dx^\mu dx^\nu \;=\;&
      (1-y)\,\delta g_{yy}^{(2)}\,d\tau^{2}+ 
      \frac{\delta g_{yy}^{(2)}}
           {(1-y)(1+y)^{2}}\,dy^{2}\;,
\end{aligned}
\end{equation}
and
\begin{equation}
\label{Eq: Ansatz:deltaF xy}
    \delta F= \delta F_{\tau y}(y)\, d\tau\wedge dy\,.
\end{equation}

Given the above ingredients, one can follow the steps in Section \ref{Sec: cancellation} and find the $\tfrac{3}{2}\log{T}$ correction to the thermodynamic entropy.

\section{Kerr-dS$_4$ black holes as an example}
\label{Sec: Kerr_dS}
In this section, we provide an explicit example in which our theorem can be applied. We also clarify the role of ensemble choices in the selection of the one-parameter  type A family of solutions in the  one-loop correction of the thermodynamic entropy. 

We turn off the matter field and choose a positive cosmological constant $\Lambda=3/\ell^2$ in \eqref{Eq: Maxwell action}. The stationary and axial-symmetric solution in four dimensions is the Kerr-dS$_4$ black hole. Relevant studies on non-rotating asymptotic dS black holes appeared in \cite{Maulik:2025phe,Blacker:2025zca}. Here we consider the rotating case. The Kerr-dS$_4$ solution is given by
\begin{equation}
    ds^2
= -\frac{\Delta_r}{\rho^2}
  \Bigl(\dd t - \frac{a}{\Theta}\sin^2\theta\,\dd \phi\Bigr)^{2}
+ \frac{\rho^2}{\Delta_r}\,\dd r^2
+ \frac{\rho^2}{\Delta_\theta}\,\dd \theta^2
+ \frac{\Delta_\theta}{\rho^2}\,\sin^2\theta
  \Bigl(a\,\dd t - \frac{r^2 + a^2}{\Theta}\,\dd \phi\Bigr)^{2}
\end{equation}
with parameters
\begin{equation}
\begin{aligned}
\Delta_r       &= (r^2 + a^2)\Bigl(1 - \frac{r^2}{\ell^2}\Bigr) \;-\; 2 M r, 
&\quad
\Delta_\theta &= 1 + \frac{a^2}{\ell^2}\cos^2\!\theta,\\
\rho^2         &= r^2 + a^2\cos^2\!\theta,
&\quad
\Theta        &= 1 + \frac{a^2}{\ell^2}\,.
\end{aligned}
\end{equation}
This solution has three parameters $\{M,a,l\}$. It is also convenient to parameterize this solution by the radii of its inner, outer, and cosmological horizons $r_h=\{r_-,r_+,r_c\}$. In order to do so, let us first express $\Delta_r$ by the radii,
\begin{equation}
\label{Eq: Delta_r Kerr dS}
    \Delta_r=- \frac{1}{\ell^2}\,(r + r_- + r_+ + r_c)\,(r - r_-)\,(r - r_+)\,(r - r_c).
\end{equation}
The original parameters in the metric can be expressed by the radii,
\begin{equation}
    \begin{aligned}
        a^2=&l^2-r_c^2-r_+^2-r_-^2-r_cr_+-r_cr_--r_+r_-,\\
        M=&\frac{1}{2l^2}(r_c+r_+)(r_c+r_-)(r_++r_-),\\
        l^2=&\frac{1}{2}\Bigl(
  r_{c}^{2} + r_{c}r_{-} + r_{-}^{2}
  + r_{c}r_{+} + r_{-}r_{+} + r_{+}^{2}
  + \\&\sqrt{%
      4\,r_{c}r_{-}r_{+}\,(r_{c}+r_{-}+r_{+})
      + \bigl(r_{c}^{2} + r_{+}^{2}+ r_{-}^{2} + r_{-}r_{+} + r_{c}r_{-}+r_{c}r_{+}\bigr)^{2}%
    }
\Bigr).
    \end{aligned}
\end{equation}
For each horizon, we can define the associated entropy and temperature by 
\begin{equation}
    S_h=\pi r_h^2,\quad T_h=\frac{1}{4\pi}|\Delta_r'(r_h)|. 
\end{equation}
Similarly to the Reissner-Nordstrom-de Sitter (RNdS$_4$) in \cite{Maulik:2025phe,Blacker:2025zca}, the physical region for parameters $a$ and $M$ with specific $l$ has a ``shark fin" shape shown in Figure \ref{fig:shark finn}.
\begin{figure}
    \centering
    \includegraphics[width=0.8
    \linewidth]{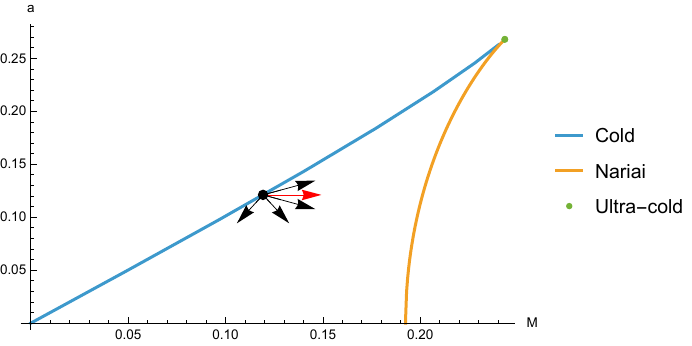}
    \caption{Phase space of Kerr-dS$_4$ black hole: We set $\ell=1$ in this plot. The arrows represent different ensemble choices, the red arrow is the choice used in Section \ref{Sec: Near-cold}.}
    \label{fig:shark finn}
\end{figure}
Two of the boundaries and their point of intersection correspond to the three different types of extremal solutions:
\begin{itemize}
    \item \textbf{Cold}: $r_-=r_+<r_c$.
    \item \textbf{Nariai}: $r_-<r_+=r_c$.
    \item \textbf{Ultra-cold}: $r_-=r_+=r_c$.
\end{itemize}
Expression \eqref{Eq: Delta_r Kerr dS} highlights the fundamental difference between cold or Nariai and ultra-cold black holes. For both cold and Nariai black holes, the radius of the horizon $r_0=r_-=r_+$ or $r_0=r_+=r_c$ is a double root of $\Delta_r$, while in the ultra-cold case, $r_0=r_-=r_+=r_c$ is a triple root. The cold and Nariai limits are examples of the \textit{quadratic extremality} defined before in Definition \ref{def: 1}, while the ultra cold limit is not. As pointed out in \cite{Maulik:2025phe,Blacker:2025zca}, the near-horizon geometry of ultra-cold black holes does not have an AdS$_2$ structure. We will discuss the issue of ultra-cold black holes later in this section, as well as some subtleties in the treatment  of the Nariai black holes. For now, let us first focus on our main goal, namely, to provide an explicit application of our theorem. We use the simplest extremality, the cold limit .

\subsection{Near-cold black holes}
\label{Sec: Near-cold}
In studying the near-horizon geometry of near-extremal black holes, it is convenient to use $\{\ell,r_0,T\}$ as parameters. Other parameters will be expressed as a series of $T$ with coefficients containing $\ell$ and $r_0$:
\begin{align}
\label{Eq: Kerr dS constant a-a}
a &= \frac{\sqrt{\ell^{2} r_{0}^{2}-3 r_{0}^{4}}}{\sqrt{\ell^{2}+r_{0}^{2}}}, \\[4pt]
M &= \frac{ r_{0}\,(-\ell^{2}+r_{0}^{2})^{2}}{\ell^{2}(\ell^{2}+r_{0}^{2})}
    + \frac{8\,\ell^{6}\pi^{2} r_{0}^{3}(\ell^{2}+r_{0}^{2})\,T^{2}}
           {(\ell^{2}+3 r_{0}^{2})^{2}\,(\ell^{4}-6\ell^{2} r_{0}^{2}-3 r_{0}^{4})}
    + \mathcal{O}(T^{3}), \label{Eq: Kerr dS constant a-M} \\[4pt]
r_{+} &= r_{0}
    + \frac{4\,\ell^{4}\pi\, r_{0}^{2}(\ell^{2}+r_{0}^{2})\,T}
           {(\ell^{2}+3 r_{0}^{2})\,(\ell^{4}-6\ell^{2} r_{0}^{2}-3 r_{0}^{4})}
    + \mathcal{O}(T^{2}).  \label{Eq: Kerr dS constant a-rp}
\end{align}
Equations \eqref{Eq: Kerr dS constant a-a}-\eqref{Eq: Kerr dS constant a-rp} explicitly determine the one-parameter family of type A introduced in Definition \ref{def: 3}. 
Note that we have chosen $a$ to stay constant and independent of the temperature, $T$, which is equivalent to a choice of ensemble. More details about ensemble choice will be discussed in the next subsection. To zoom into the near-horizon geometry, we make the following coordinate transformation
\begin{equation}
\label{eq:coord-trans}
r \to r_{+}+T k_{1}(y-1),\quad
t \to -\frac{i}{k_{2}T}\tau,\quad
\phi \to \varphi - \frac{i k_{3}}{T}\tau + i k_{4}\tau,
\end{equation}
with 
\begin{equation}
    \begin{aligned}
k_{1} &= \frac{4\,\ell^{4}\,\pi\, r_{0}^{2}\,(\ell^{2}+r_{0}^{2})}
              {(\ell^{2}+3 r_{0}^{2})\bigl(\ell^{4}-6\ell^{2} r_{0}^{2}-3 r_{0}^{4}\bigr)},\qquad 
k_{2} = 2\pi i\,\frac{\ell^{2}(\ell^{2}+r_{0}^{2})}
              {\ell^{4}+2\ell^{2} r_{0}^{2}-3 r_{0}^{4}},\\[4pt]
k_{3} &= \frac{\sqrt{\ell^{2}-3 r_{0}^{2}}\;(\ell^{2}+3 r_{0}^{2})}
              {2 k_{2}\,\ell^{2} r_{0}\,\sqrt{\ell^{2}+r_{0}^{2}}},\qquad 
k_{4} = \frac{\sqrt{\ell^{2}-3 r_{0}^{2}}\;(\ell^{2}+r_{0}^{2})(\ell^{2}+3 r_{0}^{2})}
              {\ell^{4}-6\ell^{2} r_{0}^{2}-3 r_{0}^{4}}.
\end{aligned}
\end{equation}
The Euclidean near-horizon geometry is given by \eqref{Eq: BG ansatz 4d} with the following functions:
\newcommand{\A}{\ell^{2}+r_{0}^{2}+(\ell^{2}-3r_{0}^{2})\cos^{2}\theta}
\newcommand{\B}{\ell^{4}-6\ell^{2}r_{0}^{2}-3r_{0}^{4}}
\begin{equation}
 \begin{aligned}
g_{1}(\theta) &= \frac{\ell^{2} r_{0}^{2}\bigl[\ell^{2}+r_{0}^{2}+(\ell^{2}-3 r_{0}^{2})\cos^{2}\theta\bigr]}
                    {\ell^{4}-6\ell^{2} r_{0}^{2}-3 r_{0}^{4}},\quad
g_{2}(\theta) = \frac{\ell^{2} r_{0}^{2}\bigl[\ell^{2}+r_{0}^{2}+(\ell^{2}-3 r_{0}^{2})\cos^{2}\theta\bigr]}
                    {\ell^{2}(\ell^{2}+r_{0}^{2})+r_{0}^{2}(\ell^{2}-3 r_{0}^{2})\cos^{2}\theta},\\[2pt]
g_{3}(\theta) &= \frac{4\,\ell^{2} r_{0}^{2}}
                    {(\ell^{2}+3 r_{0}^{2})^{2}\bigl[\ell^{2}+r_{0}^{2}+(\ell^{2}-3 r_{0}^{2})\cos^{2}\theta\bigr]},\quad 
k = \frac{\sqrt{\ell^{2}-3 r_{0}^{2}}\;(\ell^{2}+3 r_{0}^{2})\,\sqrt{\ell^{2}+r_{0}^{2}}}
         {\ell^{4}-6\ell^{2} r_{0}^{2}-3 r_{0}^{4}}.
\end{aligned}   
\end{equation}
The small-temperature perturbation of the metric is given by \eqref{Eq:g1+g2}-- \eqref{Eq: delta g2}. We show the explicit functions in Appendix \ref{App: Small temperature functions}. The zero modes are given by \eqref{Eq: Tensor Zero Modes}. After direct calculation, one finds that the lifted eigenvalues take the form of 
\begin{equation}
\alpha(\ell, r_0)\,\, nT
\,\quad n\ge2,
\end{equation}
where $\alpha(\ell, r_0)$ stands for a very long expression in terms of $\ell$ and $r_0$. As we showed in previous sections, the above set of eigenvalues leads to the $\frac{3}{2}\log{T}$ correction to the thermodynamic entropy.

\subsection{Ensemble choices}
In the various explicit examples considered in \cite{Maulik:2024dwq} it was noted that the final result of $\frac{3}{2}\log{T}$ was robust with respect to various ensemble choices. Here we can pin point, at a concrete technical level, the universality of the result vis-\`a-vis the ensemble choice. More precisely, we can track how the ensemble choice affects the explicit modes, \eqref{Eq:g1+g2}--\eqref{Eq: delta g2}, that drive the black holes away from extremality. A natural question is: what is the difference between the small temperature perturbations when turning on temperature in different ensembles? To answer this question, we start from one point on the ``cold" boundary of the ``shark fin" and turn on temperature in different directions (see Fig. \ref{fig:shark finn}), and explore its implications for the modes in \eqref{Eq:g1+g2}-\eqref{Eq: delta g2}. Instead of keeping $a$ fixed as in \eqref{Eq: Kerr dS constant a-a}, we allow $a$ to be also a series in temperature and keep $C_k\equiv\frac{a-a_0}{M-M_0}$ as a constant labeling different ensembles, where $a_0,\,M_0$ are the values at extremality,
\begin{equation}
    \begin{aligned}
        a &= \frac{\sqrt{\ell^{2} r_{0}^{2}-3 r_{0}^{4}}}{\sqrt{\ell^{2}+r_{0}^{2}}}+C_k\left[c_m^{(2)}\,T^2+\mathcal{O}(T^3)\right], \\[4pt]
M &= \frac{ r_{0}\,(-\ell^{2}+r_{0}^{2})^{2}}{\ell^{2}(\ell^{2}+r_{0}^{2})}
    + c_m^{(2)}\,T^2
    + \mathcal{O}(T^{3}), \\[4pt]
r_{+} &= r_{0}
    + c_+^{(1)}\,T+ \mathcal{O}(T^{2}).
    \end{aligned}
\end{equation}
Our original ensemble choice corresponds to $C_k=0$, {\it  i.e.}, keeping $a$ as a constant as the temperature is turned on. We choose five different ensembles $C_k=\{-2,\,-\frac{1}{2},\,0,\,\frac{1}{2},\,2\}$, which are shown by arrows in Fig.\ref{fig:shark finn}. All of these choices have the same near-horizon-extremal geometry $\bar{g}$. All small temperature perturbations take the form of \eqref{Eq:g1+g2}--\eqref{Eq: delta g2}. The only difference appears in the explicit expressions of $\delta g_{\tau\varphi}^{(22)}(\theta)$ which is a mode of $\delta g_2$. Near the $AdS_2$ boundary, its wavefunction density goes to zero. It does not change the asymptotics of the geometry. Therefore, it does not contribute to the lifted eigenvalues of the zero modes. It can also be seen explicitly from \eqref{Eq: Explicit lifted eigenvalues n>2} and \eqref{Eq: Explicit lifted eigenvalue n=2} that $\delta g_{\tau\varphi}^{(22)}$ does not enter the lifted eigenvalues. From this example, we clarify how the $\frac{3}{2}\log{T}$ correction is independent of the choice of ensemble.

\subsection{Near-Nariai and Near-ultra-cold black holes}
As mentioned in the remarks after the proof of \textit{Theorem \ref{the: 1}}, the lifted eigenvalues of the tensor zero modes are actually allowed to be negative, i.e., $\tilde{\Lambda}<0$. This is exactly the case of Nariai black holes. For cold black holes and all other exremal black hole examples in \cite{Maulik:2024dwq}, the radius of the horizon is a double root of $\Delta_r$:
\begin{equation}
    \Delta_r=L(r-r_0)^2+\mathcal{O}\left((r-r_0)^3\right),
\end{equation}
where $L$ is a positive constant. By taking the near-horizon limit, we end up with an AdS$_2$ geometry for the $\{\tau,y\}$ part. However, in the case of Nariai black holes, $L$ is negative. After taking the near-horizon limit, we obtain the ``-AdS$_2$" geometry mentioned in \cite{Blacker:2025zca}. This minus sign propagates to the lifted eigenvalues. With a negative eigenvalue, the Gaussian integral is not well defined. However, at least in the sense of analytic continuation, this will only influence the form of the temperature-independent term in \eqref{Eq: theorem: 3/2logT} and contribute an imaginary constant. See \cite{Maldacena:2024spf,Chen:2025jqm} for more discussion on the negativity of the state-counting partition function. \cite{Aalsma:2025lcb} also proposes that, with a suitable choice of observer, no logarithmic temperature correction is needed for Nariai black holes. As clarifying the subtlety of Nariai black holes is not the main purpose of this paper, we postpone a more detailed investigation of this subtlety to future work.

The ultra-cold black holes are not included in our theorem. As mentioned in \cite{Maulik:2025phe,Blacker:2025zca} , the near-horizon extremal geometry of an ultra-cold black hole is $Mink_2$ in Lorentzian signature, which does not satisfy \eqref{Eq: BG ansatz 4d}. Though we are not going to solve the problem of ultra-cold black holes in this paper, we try to provide a new perspective from which the appearance of $Mink_2$ is not surprising. One will see the $\textit{Mink}_2$ geometry is a natural result when we take the near-horizon limit too fast. The flat geometry is actually a reflection of a well-known fact: The near-horizon geometry of a finite temperature black hole is the Rindler spacetime. To illustrate this point, in the following, we show that even for near extremal Reissner--Nordstr\"om{} black holes with flat asymptotics, if we take the near-horizon limit too fast, a flat geometry emerges. This is also one of the reasons why we use temperature as both the ``near-horizon" and ``near-extremal" parameters simultaneously in \eqref{eq:coord-trans}. The real meaning of this limit is that the ``near-horizon" parameter and the ``near-extremal" parameter should be infinitesimals of the same order. 

We start with the Reissner--Nordstr\"om{} black hole solution of Einstein-Maxwell equation
\begin{equation}
    ds^{2}
  = -f(r)\,dt^{2}
    + \frac{dr^{2}}{f(r)}
    + r^{2}\left(d\theta^{2}+\sin^{2}\theta\,d\varphi^{2}\right),
\qquad
f(r)=1-\frac{2M}{r}+\frac{Q^{2}}{r^{2}},
\end{equation}
where \(M\) is the ADM mass and \(Q\) the (electric) charge. The gauge potential and field strength are given by
\begin{equation}
    A = -\frac{Q}{r}\,dt,\qquad F = dA = \frac{Q}{r^{2}}\,dr\wedge dt.
\end{equation}
The temperature of the black hole is given by 
\begin{equation}
    T=\frac{f'(r)}{4\pi}\big|_{r=r_+},
\end{equation}
where $r_+$ is the larger root of $f(r)=0$. Keeping  $Q$ as a constant, $Q=r_0$, we have the following temperature dependence of the mass and radius of the outer horizon:
\begin{equation}
    M=r_0+2\pi^2\,r_0^3\,T^2+\mathcal{O}(T^3),\qquad r_-=r_0-2\pi\,r_0^2\, T+\mathcal{O}(T^2).
\end{equation}
To obtain the near-horizon-extremal geometry of a black hole, there are two limits to take: $T\to0$ and $r-r_+\to 0$. They do not generally commute. To tract the difference between them, we use the following coordinate transformation
\begin{equation}
\label{eq: coord-trans, diff limits}
r \to r_{+}+\lambda\, \tilde{r}
\end{equation}
where $\lambda$ is the infinitesimal parameter that controls ``how fast" we zoom into the near-horizon region, i.e. the ``near-horizon" parameter. Then we do a double series expansion of the function $f(r)$ around $T=0,\,\lambda=0$:
\begin{equation}
\label{Eq: F expansion}
\begin{aligned}
        f(r)=&\frac{1}{r^2}(r-r_-)(r-r_+)\\
        =&\left[4\pi\tilde{r}\,T+\mathcal{O}(T^2)\right]\lambda+\left[\frac{\tilde{r}^2}{r_0^2}+\mathcal{O}(T)\right]\lambda^2+...
\end{aligned}
\end{equation}
In taking the near-extremal limit $T\to 0$ and the near-horizon limit $\lambda\to 0$, we have the following three options:
\begin{itemize}
    \item $T\to0,\,\lambda\to0,\,\frac{T}{\lambda}\to0$,
    \item $T\to0,\,\lambda\to0,\,\frac{T}{\lambda}\to {\rm Finite}$,
    \item $T\to0,\,\lambda\to0,\,\frac{T}{\lambda}\to\infty$.
\end{itemize}
Let us discuss these limits respectively.

\paragraph{Near-horizon exact extremal limit:} Taking the first choice, $\frac{T}{\lambda}\to 0$ is actually equivalent to taking the near-horizon limit $\lambda\to0$ at exact extremality $T=0$. The leading contribution in \eqref{Eq: F expansion} is 

\begin{equation}
    f(r)\approx\lambda^2\frac{\tilde{r}^2}{r_0^2}.
\end{equation}
The whole metric is 
\begin{equation}
    ds^2=-\lambda^2\frac{\tilde{r}^2}{r_0^2}dt^2+\frac{r_0^2\,d\tilde{r}^2}{\tilde{r}^2}+r_0^2(d\theta^2+\sin^2\theta\,d\phi^2)
\end{equation}
By proper rescaling, one can find this is actually the well known $AdS_2\times S^2$ throat with $AdS_2$ in the Poincar\'e coordinates.

\paragraph{Near-horizon-near-extremal limit:} For the second choice, $\frac{T}{\lambda}\to Finite$, we directly choose $\lambda=T$ for convenience. In this case, the two terms in \eqref{Eq: F expansion} are of the same order. We keep both of them:
\begin{equation}
f(r)\approx\lambda^2\left(4\pi\tilde{r}+\frac{\tilde{r}^2}{r_0^2}\right)
\end{equation}
By coordinate transformation $\tilde{r}\to k_1\,y+k_2$, with proper $k_1,\,k_2$, then rescaling $t$, one finds the following metric
\begin{equation}
    ds^2=-(y^2-1)dt^2+\frac{dy^2}{y^2-1}+r_0^2(d\theta^2+\sin^2\theta\,d\phi^2).
\end{equation}
This limit is actually what was used in the previous sections. This is also a $AdS_2\times S^2$ geometry, but in the ``black hole" coordinates. The advantage of this limit is twofold: First, in this limit, it is convenient to produce the small temperature perturbations, which are used to regulate the zero modes mentioned before. Second, this limit avoids the subtlety of working on the background of an exact extremal black hole, which we already know is not a well-defined saddle point because of large quantum fluctuations. This black hole patch describes the near-horizon geometry of a black hole with a small but nonzero temperature. The temperature of the original black hole can be extracted from the period of Euclidean time $\tau=i t$. 

\paragraph{Near-horizon of finite temperature black holes:} The last choice, $\frac{T}{\lambda}\to \infty$, means that we take the near-horizon limit ``faster" than the near-extremal limit. This leads to 
\begin{equation}
  f(r)\approx 4\pi\,\tilde{r}\,\lambda\, T.  
\end{equation}
The metric is given by
\begin{equation}
    ds^2=-4\pi\,\tilde{r}\,\lambda \,T \,dt^2+\frac{d\tilde{r}^2}{4\pi\,\tilde{r}\,\lambda\,T}+r_0^2\,(d\theta^2+\sin^2\theta\,d\phi^2).
\end{equation}
By coordinate transformation
\begin{equation}
    \tilde{r}\to\pi\,\lambda\,T\,\rho^2,\qquad t\to\frac{d\tilde{t}}{2\pi\,\lambda\,T},
\end{equation}
the metric takes the form of $Rind_2\times S^2$:
\begin{equation}
    ds^2=-\rho^2\,d\tilde{t}^2+d\rho^2+r_0^2\,(d\theta^2+\sin^2\theta\,d\phi^2).
\end{equation}
This reflects the well-known fact that the near-horizon geometry of a finite temperature black hole is the Rindler spacetime. This provides a possible explanation for why we encounter $Mink_2\times S^2$ when zooming into the near-horizon region of RN--dS$_4$ black holes  \cite{Maulik:2025phe,Blacker:2025zca}.


\section{Conclusions}
In this manuscript, we have considered a general theory of Einstein gravity coupled to Abelian vector fields and neutral scalars with arbitrary scalar potential. For  this class of theories, we have established four lemmas and a theorem that demonstrate the universality of the tensor modes logarithmic-in-temperature contribution  to the thermodynamic entropy: $\delta S= \frac{3}{2}\log T$. \textit{Lemma \ref{lem: 1}} and \textit{Lemma \ref{lem: 2}} characterize the general form of the near-horizon extremal configurations. \textit{Lemma \ref{lem: 3}} states the universal existence of a set of tensor zero modes. The structure of the zero modes, in particular their origin in large diffeomorphisms, creates a direct identification with the well-known Schwarzian modes in the context of near-horizon extremal black hole geometries. \textit{Lemma \ref{lem: 4}} systematizes a crucial simplification of the Linchnerowicz operator, which implies that the contribution of the tensor zero modes to the partition function is independent of the  matter sector. \textit{Assumption \ref{ass: 2}} helps characterize the way black holes are driven away from extremality. All these lemmas and assumptions work for any spacetime dimension greater or equal to four and include asymptotically flat, AdS and dS asymptotics. Finally, by explicit tensor calculations, we prove in \textit{Theorem \ref{the: 1}} that the $\tfrac{3}{2}\log{T}$ correction to the entropy is universal in dimensions four, five, and six. We then apply our findings to the  explicit example of the cold Kerr-dS$_4$ black hole. We show the role of ensemble choice and provide some comments on the subtlety of Nariai and ultra-cold limits of asymptotically dS black holes. We anticipate that increasingly more involved calculations are required to establish similar results in even higher dimensional black holes.

More broadly, the current manuscript expands and formalizes observations made for asymptotically flat and asymptotically AdS black holes in \cite{Maulik:2024dwq}.  We have clarified the underlying conditions leading to the universal tensor contribution to the one-loop level gravitational partition function. Our analysis can be seen as a low-energy discussion and can, in fact, be connected to low-energy theory living in the throat region which includes Schwarzian modes. The higher-dimensional nature of our computations gives us hope that whatever the final form of the theory (full gravitational path integral), the role of symmetries and its impact via zero modes are likely to remain part of the picture.

We leave open the issue of the fate of the zero mode in the full spacetimes; in this manuscript we have focused on the one-loop contribution from the near-horizon region. There has been gathering evidence that the zero modes are a result of taking the low-temperature limit \cite{Kolanowski:2024zrq}. Arguably, the best evidence arises from the series of work \cite{Bonelli:2021uvf,Arnaudo:2024rhv,Arnaudo:2024bbd,Arnaudo:2025btb}. Given the prominence of this issue, we hope to bridge the two analyses and explore beyond the particular results stated here. Another interesting direction would be to pursue potential connections to recent work aiming to clarify the gravitational path integral in asymptotically de-Sitter spacetimes \cite{Ivo:2025yek,Turiaci:2025xwi}.

\section*{Acknowledgments}
We are thankful to Giulio Bonelli, Roberto Empar\'an, Sabare Jayaprakash, Jim Liu, Sabyasachi Maulik, Augniva Ray, Alessandro Tanzini. This work is partially supported by the U.S. Department of Energy under grant DE-SC0007859. JZ was supported in part by a Leinweber Summer Graduate Fellowship.

\appendix

\section{Coordinate transformations}
\label{App: Coordinate transformation}
In this appendix we show the coordinate transformations that connect \eqref{Eq: gbar Kunduri} and \eqref{Eq: BG ansatz}. Let us first focus on the $dv^2,dv\,dr$ parts. We first perform the following transformation to move from the null coordinates to ``Poincare coordinates of $AdS_2$", 
\begin{equation}
    v\to T+\frac{1}{R},\quad r\to\frac{R}{A_0}.
\end{equation}
This gives
\begin{equation}
    ds_2^2\equiv\Gamma(\theta_n)[A_0 r^{2} dv^{2} + 2 \,dv\,dr]=-\frac{\Gamma(\theta_n)}{A_0}[ R^{2} dT^{2} + \frac{dR^2}{R^2} ].
\end{equation}
We perform the following transformation 
\begin{equation}
    R\to y-\Delta\cosh{t},\quad T\to\frac{\Delta \sinh{t}}{y-\Delta\cosh{t}},
\end{equation}
where $\Delta\equiv\sqrt{y^2-1}$. This brings us to the coordinates ${t,y}$ that we need in \eqref{Eq: BG ansatz}:
\begin{equation}
ds_2^2=g_1(\theta_n)[ -(y^2-1) dt^{2} + \frac{dy^2}{(y^2-1)^2} ],
\end{equation}
where we define $g_1(\theta_n)\equiv-\tfrac{\Gamma(\theta_n)}{A_0}$.
Now we turn to the $d\varphi_id\varphi_j$ and $d\varphi_idv$ parts. Under the above coordinate transformations,
\begin{equation}
\begin{aligned}
&(d\varphi_i+k^i r dv)\\
&\to\left(d\varphi_i+\frac{
k^i \left( -\Delta + y \cosh{t} + \sinh{t} \right)
}{
A_0 \left( y \Delta + \cosh{t} - y^2 \cosh{t} \right)
}dy\,+\frac{
k^i \left( -1 + y^2 \right) \left( -\Delta + y \cosh{t} + \sinh{t} \right)
}{
A_0 \left( y \Delta + \cosh{t} - y^2 \cosh{t} \right)
}\right).
\end{aligned}
\end{equation}
In order to eliminate the $dy$ term, we introduce
\begin{equation}
    \Phi(t,y)\equiv- \frac{2 \Delta \, \arctan \left( \frac{y - \tanh\frac{t}{2}}{\sqrt{-1 - y} \, \sqrt{-1 + y}} \right)}{A_0 \sqrt{-1 - y} \, \sqrt{-1 + y}}
- \frac{1}{A_0}  \log\left( y + \Delta \right) -\frac{1}{A_0}\log\left(  y \sinh{t}  -\cosh{t}\right),
\end{equation}
and perform coordinate transformation:
\begin{equation}
    \varphi_i\to\tilde{\varphi}_i-k^i\Phi(t,y)+\frac{k^i}{A_0}t.
\end{equation}
After replacing $k^i\to A_0\, \tilde{k}_i$, we find
\begin{equation}
\label{Eq: phi t coordinate Trans}
    (d\varphi_i+k^i r dv)\to\left(d\tilde{\varphi}_i-\tilde{k}_i(y-1)dt\right).
\end{equation}
The last step is to rotate to Euclidean time $t\to-i\tau$, after which one finds \eqref{Eq: BG ansatz}, where we have eliminated the tildes for simplicity of notation. 

Performing the same coordinate transformation to \eqref{Eq: F ansatz Kunduri}, one finds
\begin{equation}
    F^{I}_{v r}(\theta_n)\,dv\wedge dr\to \frac{iF^{I}_{v r}(\theta_n)}{A_0}\,d\tau\wedge dy
\end{equation}
Defining $F^{I}_{\tau y}(\theta_n)\equiv\frac{iF^{I}_{v r}(\theta_n)}{A_0}$ and applying \eqref{Eq: phi t coordinate Trans}, we find \eqref{Eq: Ansatz F} from \eqref{Eq: F ansatz Kunduri}. Again, we have omitted the tildes for simplicity.

\section{The small temperature perturbation functions}
\label{App: Small temperature functions}
This appendix shows explicit expressions of functions in the small temperature perturbations \eqref{Eq:g1+g2}--\eqref{Eq: delta g2} in the near-cold example in Section \ref{Sec: Near-cold}.

\begin{equation}
\begin{aligned}
\delta g_{\tau\tau}(\theta)
=&\Bigl(\ell^{6} \pi r_{0}^{3} (\ell^{4} - r_{0}^{4})^{2}
\bigl(
-13 \ell^{4} + 46 \ell^{2} r_{0}^{2} + 107 r_{0}^{4}
- 4 \left(5 \ell^{4} - 6 \ell^{2} r_{0}^{2} - 27 r_{0}^{4}\right)\cos(2\theta)
\\&+(\ell^{2} - 3 r_{0}^{2})^{2} \cos(4\theta)
\bigr)\Bigr)
/\left((\ell^{2} + 3 r_{0}^{2})
(\ell^{4} - 6 \ell^{2} r_{0}^{2} - 3 r_{0}^{4})^{3}
\left(\ell^{2} + r_{0}^{2} + (\ell^{2} - 3 r_{0}^{2}) \cos^{2}\theta\right)^{2}\right)
, \\[6pt]
\delta g_{yy}^{(1)}(\theta)
=& \frac{8 \ell^{6} \pi r_{0}^{3} (\ell^{2} + r_{0}^{2})^{2}
\left(\ell^{4} - 3 \ell^{2} r_{0}^{2} - 4 r_{0}^{4}
+ \ell^{2} r_{0}^{2} \cos(2\theta) - 3 r_{0}^{4} \cos(2\theta)\right)}
{(\ell^{2} + 3 r_{0}^{2}) \left(\ell^{4} - 6 \ell^{2} r_{0}^{2} - 3 r_{0}^{4}\right)^{3}}
, \\[6pt]
\delta g_{yy}^{(2)}(\theta)
=& -\,\frac{
4 \ell^{6} \pi r_{0}^{3} (\ell^{2} + r_{0}^{2})^{2}
\left(\ell^{4} - 2 \ell^{2} r_{0}^{2} - 7 r_{0}^{4}\right)
\left(3 \ell^{2} - r_{0}^{2} + (\ell^{2} - 3 r_{0}^{2}) \cos(2\theta)\right)
}{
\left(\ell^{4} - 6 \ell^{2} r_{0}^{2} - 3 r_{0}^{4}\right)^{3}
\left(\ell^{4} + 2 \ell^{2} r_{0}^{2} - 3 r_{0}^{4}\right)
}
, \\[6pt]
\delta g_{\theta\theta}(\theta)
=& \frac{
8 \ell^{6} \pi r_{0}^{3} (\ell^{2} + r_{0}^{2})^{2}
}{
\left(\ell^{4} - 6 \ell^{2} r_{0}^{2} - 3 r_{0}^{4}\right)
\left(\ell^{6} + 4 \ell^{4} r_{0}^{2} + 3 \ell^{2} r_{0}^{4}
      + r_{0}^{2}(\ell^{4} - 9 r_{0}^{4}) \cos^{2}\theta\right)
}
, \\[6pt]
\delta g_{\varphi\varphi}(\theta)
=&\Bigl(
8 \ell^{6} \pi r_{0}^{3} (\ell^{2} + r_{0}^{2})^{2}
\left(\ell^{2} + r_{0}^{2} + (\ell^{2} - 3 r_{0}^{2}) \cos(2\theta)\right)
\bigl(2 \ell^{4} + 3 \ell^{2} r_{0}^{2} - 3 r_{0}^{4}
+ r_{0}^{2}(\ell^{2} - 3 r_{0}^{2}) \\&\cos(2\theta)\bigr)
\sin^{2}\theta
\Bigr)/\Bigl(
(\ell^{2} + 3 r_{0}^{2})^{3}
\left(\ell^{6} - 7 \ell^{4} r_{0}^{2} + 3 \ell^{2} r_{0}^{4} + 3 r_{0}^{6}\right)
\left(\ell^{2} + r_{0}^{2} + (\ell^{2} - 3 r_{0}^{2}) \cos^{2}\theta\right)^{2}
\Bigr)
, \\[6pt]
\delta g_{\tau\varphi}^{(1)}(\theta)
=& -\,\Bigl(
i\,\ell^{6}\pi r_{0}^{3}
\left(\ell^{2}-3r_{0}^{2}\right)^{3/2}
\left(\ell^{2}+r_{0}^{2}\right)^{5/2}
\bigl(
27\ell^{2}-33r_{0}^{2}
+4\left(9\ell^{2}-7r_{0}^{2}\right)\cos(2\theta)
\\&
+\left(\ell^{2}-3r_{0}^{2}\right)\cos(4\theta)
\bigr)\sin^{2}\theta
\Bigr)
/\Bigl(
2\left(\ell^{2}+3r_{0}^{2}\right)^{2}
\left(\ell^{4}-6\ell^{2}r_{0}^{2}-3r_{0}^{4}\right)^{2}
\left(\ell^{2}+r_{0}^{2}
+\left(\ell^{2}-3r_{0}^{2}\right)\cos^{2}\theta\right)^{2}
\Bigr)
, \\[6pt]
\delta g_{\tau\varphi}^{(21)}(\theta)
=& \Bigl(
i\,\ell^{6}\pi r_{0}^{3}
\left(\ell^{2}-3r_{0}^{2}\right)^{3/2}
\left(\ell^{2}+r_{0}^{2}\right)^{5/2}
\bigl(
27\ell^{2}-33r_{0}^{2}
+4\left(9\ell^{2}-7r_{0}^{2}\right)\cos(2\theta)
\\&
+\left(\ell^{2}-3r_{0}^{2}\right)\cos(4\theta)
\bigr)\sin^{2}\theta
\Bigr)
/\Bigl(
2\left(\ell^{2}+3r_{0}^{2}\right)^{2}
\left(\ell^{4}-6\ell^{2}r_{0}^{2}-3r_{0}^{4}\right)^{2}
\left(\ell^{2}+r_{0}^{2}
+\left(\ell^{2}-3r_{0}^{2}\right)\cos^{2}\theta\right)^{2}
\Bigr)
, \\[6pt]
\delta g_{\tau\varphi}^{(22)}(\theta)
=& \Bigl(
4\,i\,\ell^{6}\pi r_{0}^{3}
\sqrt{\ell^{2}-3r_{0}^{2}}\,
\left(\ell^{2}+r_{0}^{2}\right)^{3/2}
\Bigl(
-9\ell^{10}-35\ell^{8}r_{0}^{2}
+210\ell^{6}r_{0}^{4}
+198\ell^{4}r_{0}^{6}-9\ell^{2}r_{0}^{8}
\\&
-99r_{0}^{10}
+\left(\ell^{2}-3r_{0}^{2}\right)
\left(\ell^{8}-18\ell^{6}r_{0}^{2}
+40\ell^{4}r_{0}^{4}
+66\ell^{2}r_{0}^{6}
+39r_{0}^{8}\right)\cos(2\theta)
\Bigr)\sin^{2}\theta
\Bigr)
\\&/\Bigl(
(\ell-r_{0})(\ell+r_{0})
\left(\ell^{2}+3r_{0}^{2}\right)^{2}
\left(\ell^{4}-6\ell^{2}r_{0}^{2}-3r_{0}^{4}\right)^{3}\bigl(
3\ell^{2}-r_{0}^{2}
+\left(\ell^{2}-3r_{0}^{2}\right)\cos(2\theta)
\bigr)
\Bigr)
 .
\end{aligned}
\end{equation}

\section{The lifted eigenvalues}
\label{App: Lifted eigenvalue}
Here we show the explicit dependence of the lifted eigenvalues on the background and perturbed configurations.
\begin{equation}
\label{Eq: Explicit lifted eigenvalues n>2}
\begin{aligned}
\Lambda_n(\theta)=&\frac{1}{8\,(-1+n^{2})\,g_{1}(\theta)^{3}\,g_{2}(\theta)^{2}\,g_{3}(\theta)}\\
&\Big(
-2\,g_{2}(\theta)\,g_{3}(\theta)
\Big[\delta g_{tt}(\theta) + \delta g_{yy}^{(1)}(\theta) 
- k\big( 2i\,\delta g_{t\varphi}^{(1)}(\theta) + k\,\delta g_{\varphi\varphi}(\theta) \big) \Big]
\big(g_{1}'(\theta)\big)^{2} \\[0.4em]
&\quad + g_{1}(\theta)\Big\{
g_{1}'(\theta)\Big[
g_{2}(\theta)\Big(\delta g_{tt}(\theta) + \delta g_{yy}^{(1)}(\theta)
- k\big( 2i\,\delta g_{t\varphi}^{(1)}(\theta) + k\,\delta g_{\varphi\varphi}(\theta) \big)\Big)g_{3}'(\theta) \\[0.4em]
&\qquad + g_{3}(\theta)\Big(
-\delta g_{tt}(\theta) g_{2}'(\theta)
+ 2ik\,\delta g_{t\varphi}^{(1)}(\theta) g_{2}'(\theta)
- \delta g_{yy}^{(1)}(\theta) g_{2}'(\theta)
+ k^{2} \delta g_{\varphi\varphi}(\theta) g_{2}'(\theta) \\[0.4em]
&\qquad\quad + 2\,g_{2}(\theta)(\delta g_{tt})'(\theta)
- 4ik\,g_{2}(\theta)(\delta g_{t\varphi}^{(1)})'(\theta)
+ 2\,g_{2}(\theta)(\delta g_{yy}^{(1)})'(\theta)
- 2k^{2} g_{2}(\theta)(\delta g_{\varphi\varphi})'(\theta)
\Big)\Big] \\[0.4em]
&\qquad + 2\,g_{2}(\theta) g_{3}(\theta)
\Big[\delta g_{tt}(\theta) + \delta g_{yy}^{(1)}(\theta) 
- k\big( 2i\,\delta g_{t\varphi}^{(1)}(\theta) + k\,\delta g_{\varphi\varphi}(\theta) \big) \Big]
g_{1}''(\theta)
\Big\} \\[0.4em]
&\quad - g_{1}(\theta)^{2} \Big\{
2\,g_{2}(\theta)^{2}\,\delta g_{\varphi\varphi}(\theta) 
- g_{3}(\theta) g_{2}'(\theta)
\Big[(\delta g_{tt})'(\theta) + (\delta g_{yy}^{(1)})'(\theta) 
- k\big( 2i(\delta g_{t\varphi}^{(1)})'(\theta) + k(\delta g_{\varphi\varphi})'(\theta) \big) \Big] \\[0.4em]
&\qquad + g_{2}(\theta)\Big[
g_{3}'(\theta)
\Big((\delta g_{tt})'(\theta) + (\delta g_{yy}^{(1)})'(\theta) 
- k\big( 2i(\delta g_{t\varphi}^{(1)})'(\theta) + k(\delta g_{\varphi\varphi})'(\theta) \big) \Big) \\[0.4em]
&\qquad\quad + 2\,g_{3}(\theta)\Big(
\delta g_{\theta\theta}(\theta) 
+ (\delta g_{tt})''(\theta) 
+ (\delta g_{yy}^{(1)})''(\theta) 
- k\big( 2i(\delta g_{t\varphi}^{(1)})''(\theta) + k(\delta g_{\varphi\varphi})''(\theta) \big)
\Big)\Big]
\Big\}
\Big)
\end{aligned}
\end{equation}

\begin{equation}
\label{Eq: Explicit lifted eigenvalue n=2}
\begin{aligned}
\Lambda_2=&\frac{1}{24\,g_{1}(\theta)^{3}\,g_{2}(\theta)^{2}\,g_{3}(\theta)}\\
&\Big(
-2\,g_{2}(\theta)\,g_{3}(\theta)
\Big[\delta g_{tt}(\theta) + \delta g_{yy}^{(1)}(\theta)
- k\big( 2\,i\,\delta g_{t\varphi}^{(1)}(\theta) + k\,\delta g_{\varphi\varphi}(\theta) \big) \Big]
\big(g_{1}'(\theta)\big)^{2} \\[0.4em]
&\quad + g_{1}(\theta)\Big\{
3\,g_{2}(\theta)^{2} g_{3}(\theta)
\Big[ -\delta g_{tt}(\theta) + 2\,ik\,\delta g_{t\varphi}^{(1)}(\theta)
+ \delta g_{yy}^{(1)}(\theta) + \delta g_{yy}^{(2)}(\theta)
+ k^{2} \delta g_{\varphi\varphi}(\theta) \Big] \\[0.4em]
&\qquad - g_{3}(\theta)
\Big[ \delta g_{tt}(\theta) + \delta g_{yy}^{(1)}(\theta)
- k\big( 2\,i\,\delta g_{t\varphi}^{(1)}(\theta) + k\,\delta g_{\varphi\varphi}(\theta) \big) \Big]
g_{1}'(\theta) g_{2}'(\theta) \\[0.4em]
&\qquad + g_{2}(\theta)\Big(
g_{1}'(\theta)\Big[
\delta g_{tt}(\theta) g_{3}'(\theta)
- 2\,ik\,\delta g_{t\varphi}^{(1)}(\theta) g_{3}'(\theta)
+ \delta g_{yy}^{(1)}(\theta) g_{3}'(\theta)
- k^{2} \delta g_{\varphi\varphi}(\theta) g_{3}'(\theta) \\[0.4em]
&\qquad\quad + 2\,g_{3}(\theta)(\delta g_{tt})'(\theta)
- 4\,ik\,g_{3}(\theta)(\delta g_{t\varphi}^{(1)})'(\theta)
+ 2\,g_{3}(\theta)(\delta g_{yy}^{(1)})'(\theta)
- 2k^{2} g_{3}(\theta)(\delta g_{\varphi\varphi})'(\theta)
\Big] \\[0.4em]
&\qquad\quad + 2\,g_{3}(\theta)
\Big[\delta g_{tt}(\theta) + \delta g_{yy}^{(1)}(\theta)
- k\big( 2\,i\,\delta g_{t\varphi}^{(1)}(\theta) + k\,\delta g_{\varphi\varphi}(\theta) \big) \Big]
g_{1}''(\theta)
\Big)
\Big\} \\[0.4em]
&\quad - g_{1}(\theta)^{2}\Big\{
2\,g_{2}(\theta)^{2} \delta g_{\varphi\varphi}(\theta)
- g_{3}(\theta) g_{2}'(\theta)
\Big[(\delta g_{tt})'(\theta) + (\delta g_{yy}^{(1)})'(\theta)
- k\big( 2\,i(\delta g_{t\varphi}^{(1)})'(\theta) + k(\delta g_{\varphi\varphi})'(\theta) \big) \Big] \\[0.4em]
&\qquad + g_{2}(\theta)\Big[
g_{3}'(\theta)
\Big[(\delta g_{tt})'(\theta) + (\delta g_{yy}^{(1)})'(\theta)
- k\big( 2\,i(\delta g_{t\varphi}^{(1)})'(\theta) + k(\delta g_{\varphi\varphi})'(\theta) \big) \Big] \\[0.4em]
&\qquad\quad + 2\,g_{3}(\theta)\Big(
\delta g_{\theta\theta}(\theta)
+ (\delta g_{tt})''(\theta) + (\delta g_{yy}^{(1)})''(\theta)
- k\big( 2\,i(\delta g_{t\varphi}^{(1)})''(\theta) + k(\delta g_{\varphi\varphi})''(\theta) \big)
\Big)
\Big]
\Big\}
\Big)
\end{aligned}
\end{equation}

\bibliographystyle{JHEP}

\providecommand{\href}[2]{#2}\begingroup\raggedright\endgroup

\end{document}